\documentclass[aps,pra,superscriptaddress,superscriptaddress,longbibliography]{revtex4-2}
\usepackage{braket}
\usepackage{graphicx}
\usepackage{amsmath,amssymb}    
\usepackage{amsfonts}
\usepackage[hidelinks]{hyperref}
\usepackage{lipsum}
\usepackage{xcolor}
\usepackage{quantikz}
\usepackage{bbold}
\usepackage{adjustbox}
\usepackage{multirow}
\usepackage{color, colortbl}
\usepackage{rotating}
\usepackage{array,mathtools,amssymb,booktabs}
\usepackage[T1]{fontenc}
\usepackage[tiny, center, uppercase]{titlesec}
\usepackage[caption=false]{subfig}
\usepackage[version=3]{mhchem} 
\usepackage{amsthm}
\usepackage{thmtools}
\usepackage{thm-restate}

\theoremstyle{definition}
\newtheorem{definition}{Definition}[section]

\definecolor{Gray}{gray}{0.9}
\newcolumntype{a}{>{\columncolor{white}}c}
\newcolumntype{b}{>{\columncolor{Gray}}c}

\begin{document}

\title{Assessing the query complexity limits of quantum phase estimation\\ using symmetry aware spectral bounds}

\newcommand{\BI}{\affiliation{
Quantum Lab, Boehringer Ingelheim, 55218 Ingelheim am Rhein, Germany}}

\newcommand{\BIMedChem}{\affiliation{
Boehringer Ingelheim Pharma GmbH \& Co KG, Birkendorfer Strasse 65, 88397 Biberach, Germany}}

\author{Cristian L. Cortes}
\affiliation{QC Ware Corporation, Palo Alto, California 94301, USA}

\author{Dario Rocca}
\affiliation{QC Ware Corporation, Palo Alto, California 94301, USA}

\author{J\'er\^ome F. Gonthier}
\affiliation{QC Ware Corporation, Palo Alto, California 94301, USA}

\author{Pauline J. Ollitrault}
\affiliation{QC Ware Corporation, Palo Alto, California 94301, USA}

\author{Robert M.~Parrish}
\affiliation{QC Ware Corporation, Palo Alto, California 94301, USA}

\author{Gian-Luca R.~Anselmetti}
\BI

\author{Matthias Degroote}
\BI

\author{Nikolaj Moll}
\BI

\author{Raffaele Santagati}
\BI

\author{Michael Streif}
\BI
\date{\today}

\begin{abstract}
The computational cost of quantum algorithms for physics and chemistry is closely linked to the spectrum of the Hamiltonian, a property that manifests in the necessary rescaling of its eigenvalues. The typical approach of using the 1-norm as an upper bound to the spectral norm to rescale the Hamiltonian suits the most general case of bounded Hermitian operators but neglects the influence of symmetries commonly found in chemical systems. In this work, we introduce a hierarchy of symmetry-aware spectral bounds that provide a unified understanding of the performance of quantum phase estimation algorithms using block-encoded electronic structure Hamiltonians. We present a variational and numerically tractable method for computing these bounds, based on orbital optimization, to demonstrate that the computed bounds are smaller than conventional spectral bounds for a variety of molecular benchmark systems. We also highlight the unique analytical and numerical scaling behavior of these bounds in the thermodynamic and complete basis set limits. Our work shows that there is room for improvement in reducing the 1-norm, not yet achieved through methods like double factorization and tensor hypercontraction, but highlights potential challenges in improving the performance of current quantum algorithms beyond small constant factors through 1-norm reduction techniques alone.  
\end{abstract}
\maketitle

\section{Introduction}
Quantum computing holds the potential to significantly impact physics and chemistry, addressing various problem areas such as quantum simulation, Gibbs state sampling, solving linear systems and partial differential equations, as well as precise energy and observable estimation \cite{low2017optimal,low2019hamiltonian,gilyen2019quantum,chan2023simulating,rall2020quantum,tong2021fast,steudtner2023fault,an2023quantum,krovi2023improved,chan2023simulating,Santagati2024}. Over the past few years, a wide range of quantum algorithms have emerged, each offering distinct trade-offs in runtime and resource cost \cite{low2017optimal,low2019hamiltonian,gilyen2019quantum,martyn2021grand,dutkiewicz2022heisenberg,lin2022heisenberg,dong2022ground,li2023adaptive,ding2023even}.

The computational cost of quantum algorithms is intricately linked to the spectrum of the Hamiltonian. Many algorithms require a rescaling of the Hamiltonian to ensure that its output eigenvalues are confined within certain operational bounds, impacting the precision needed for each algorithm and thus influencing the query complexity of Hamiltonian oracles. In quantum phase estimation and its modern variants \cite{lin2022heisenberg,dong2022ground,li2023adaptive,ding2023even}, this is crucial for mitigating aliasing and potential spectral leakage effects \cite{rendon2022effects,xiong2022dual}. The scaling also emerges naturally within linear combination of unitaries (LCU) or block-encoding methods to ensure the oracle remains unitary \cite{chakraborty2018power}. The need for rescaling often transcends specific algorithms and is required across various quantum computing paradigms.

For example, within the context of block-encoding and the LCU framework, qubitization-based QPE stands out as a leading fault-tolerant algorithm for estimating the ground-state energy \cite{low2019hamiltonian,babbush2018encoding, von2021quantum,lee2021even}. To achieve a desired precision $\epsilon$, excluding the initial state overlap dependence, the total number of Toffoli gates is known to follow the asymptotic scaling, $\mathcal{O}\!\left( \lceil\tfrac{\lambda}{\epsilon}\rceil C_{\mathcal{W}[H]} \right)$, where $\lambda$ is equal to the 1-norm required for block-encoding the Hamiltonian, $C_{\mathcal{W}[H]}$ denotes the total number of Toffoli gates needed to compile the qubitization walk operator $\mathcal{W}[H]$, and $\lceil\frac{\lambda}{\epsilon}\rceil$ specifies the number of query calls, which is rounded up to the nearest integer or power of two \cite{lee2021even}. This scaling dependence highlights that reducing both $\lambda$ and $C_{\mathcal{W}[H]}$ is necessary for improving the total gate complexity of this quantum algorithm.

In electronic structure theory, these concepts have led to the development of tensor factorization techniques, which aim to provide efficient representations of the Hamiltonian \cite{motta2021low, von2021quantum, lee2021even, cohn2021quantum, oumarou2022accelerating}. For context, these methods have a long history in the development of classical algorithms for quantum chemistry, where approaches like density fitting and Cholesky decompositions are now well established \cite{dunlap1979some,werner2003fast,pedersen2009density,hohenstein2010density}. Techniques like tensor hypercontraction (THC) have also been developed over the past fifteen years to provide further enhancements \cite{parrish2013exact,lee2019systematically}. In quantum computing, methods like double factorization, regularized compressed double factorization, and THC have been used to reduce the resource cost of quantum algorithms \cite{von2021quantum,motta2021low, rubin2022compressing,lee2021even,goings2022reliably,cohn2021quantum, oumarou2022accelerating}, by minimizing either the 1-norm $\lambda$, or qubitization oracle complexity, $C_{\mathcal{W}[H]}$. It is unclear, however, whether these techniques are close to optimal, or whether there remains ample room for improvement.

Our work contributes to understanding the fundamental limits of these tensor factorization techniques both theoretically and numerically.  Specifically, we study the query complexity limits arising from lower bounds on the 1-norm $\lambda$, deferring the examination of $C_{\mathcal{W}[H]}$ for future research. To this end, we build upon the findings of Loaiza \emph{et al.}~\cite{loaiza2023reducing}, who established a spectral bound for the 1-norm $\lambda$:
\begin{restatable}[General spectral bound]{theorem}{GeneralSpectralBound}
\label{thm:GeneralSpectralBound}
For a bounded Hermitian operator $\hat H$, all of its possible LCU decompositions have an associated 1-norm $\lambda$ which is lower bounded by half the spectral range, $\Delta/2 \equiv (E_{\rm max}-E_{\rm min})/2$, 
\begin{align}
    \tfrac{1}{2}\Delta \leq \lambda,
\end{align}
where $E_{\rm max}(E_{\rm min})$ is the highest(lowest) eigenvalue of $\hat H$.
\end{restatable}

This Theorem highlights that the spectral range, $\Delta$, equal to the difference between the largest and smallest eigenvalues, acts as a robust lower bound to the 1-norm $\lambda$. The spectral range is invariant to the choice of unitary operators in the LCU decomposition, which distinguishes it from the spectral norm, $\|\hat{H}\|$, equal to the largest singular value of $\hat{H}$, which is sometimes used as a lower bound to the 1-norm. In particular scenarios, like the complete-square structure observed in double factorization \cite{von2021quantum, loaiza2023reducing}, the spectral norm can numerically overestimate the lowest achievable lower bound due to the presence of undesired identity terms that may present in the LCU decomposition.

Despite the significance of Theorem \ref{thm:GeneralSpectralBound} for bounded Hermitian operators, two fundamental concerns emerge. First, the Theorem does not take into account prior knowledge of potential symmetries in the Hermitian operator, $\hat{H}$. This omission suggests that the spectral bound, $\Delta/2$, might overestimate the spectral range for symmetry sectors relevant to many Hamiltonians and problem instances. Second, the theorem does not offer a practical method for computing such bounds numerically. Past approaches relied on exact diagonalization methods, effective only for small-scale electronic structure systems expanded with minimal basis sets \cite{loaiza2023reducing, loaiza2023block}. Consequently, evaluating these bounds for large-scale systems like FeMoco or P450, which have become standard benchmark systems for quantum computing algorithms, remains challenging.

In this work, we present a Corollary and additional Theorems that leverage the symmetries of the Hermitian operator $\hat{H}$ to establish symmetry-aware lower bounds to the 1-norm. We also propose a variational and numerically tractable approach for computing lower bounds of symmetry-sector spectral ranges, showing that  numerical estimates of symmetry-sector spectral ranges can be significantly smaller than conventional bounds for different problem instances relevant to electronic structure theory. In addition, our work analyzes the unique scaling behavior of these bounds in the thermodynamic and complete basis set limits for various molecular systems. The aim of this work is to provide a better understanding of the query complexity limits for different Hamiltonian oracles commonly used in physics and chemistry, highlighting potential room for improvement, as well as future challenges that may arise in tensor factorization and block-encoding methods.

The paper is structured as follows. In Section II, we summarize several preliminary results and definitions that will be used throughout the manuscript. In section III, we present the main theoretical results of the paper concerning the hierarchy of symmetry-aware spectral bounds, as well as theorems that establish a connection between these bounds and corresponding 1-norms of various symmetry-shifted Hamiltonians. In section IV, we present the computational procedure used for computing these bounds numerically. In section V, the numerical results are presented for large-scale benchmark systems, such as FeMoco and P450. Finally, we present numerical evidence for the scaling behavior of the symmetry-aware spectral bounds in the thermodynamic and complete basis set limits.

\section{Preliminary Definitions}
To encode the physical problem of interest, we define a Hamiltonian, $\hat{H}$, expressed as an LCU decomposition consisting of a sum of unitary operators $\hat{U}_k$ weighted by complex coefficients $c_k$,
\begin{equation}
    \hat{H} = \sum_k^L c_k \hat{U}_k,\;\;\text{where}\;\;  \lambda = \sum_k^L |c_k|,
    \label{LCU}
\end{equation}
is the 1-norm of these coefficients. Throughout this document, we also impose the requirement that the Hamiltonian $\hat{H}$ is invariant under the action of a set of symmetries, $\mathbb{S}=\{\hat{S}_m\}$, which collectively form a group $G$, satisfying the conditions $[\hat{S}_m,\hat{S}_n]=[\hat{S}_m,\hat{H}]=0,\forall\;m,n\in 1,\cdots,|\mathbb{S}|$. The group $G$ includes all individual group elements for each symmetry, as well as any combinations of these symmetries that leave the Hamiltonian invariant. Maschke's theorem enables a decomposition of the Hilbert space $\mathcal{H}$ into invariant subspaces, which are determined to be irreducible if they lack further invariant subspaces \cite{fulton2013representation,sagan2013symmetric}. Schur's lemma may then be used to reveal the block-diagonal structure of $\hat{H}$ as,
\begin{align}
    \hat{H}= \bigoplus_{\boldsymbol{m}} \bigoplus_{n=1}^{\mathrm{deg}\,\boldsymbol{m}}\hat{H}_n^{\boldsymbol{m}} \equiv \bigoplus_\mu\hat{H}_\mu,
\end{align}
where $\boldsymbol{m}$ is an index labeling the irreps of $G$ in the decomposition of $\mathcal{H}$, with multiplicity $\mathrm{deg}(\boldsymbol{m})$. Here, we adopt the condensed notation based on $\hat{H}_\mu$, representing a sub-block within a specific symmetry sector $\mu$, implicitly incorporating any multiplicities.

In the following paragraphs, we discuss the importance of the Hamiltonian spectrum, and in particular the spectral range which acts as a key parameter, affecting the performance of a broad range of quantum algorithms that employ either the Hamiltonian evolution (HE) or block-encoding (BE) oracles. Both of these Hamiltonian oracles play a key role in a wide variety of quantum algorithms with applications in preparing zero-temperature ground states, preparing finite temperature Gibbs states, estimating Green's functions and correlation functions, as well as energy and observable estimation.
\newline
\newline
\emph{Hamiltonian Evolution Oracle.} The first oracle we consider is the Hamiltonian evolution  model  defined by,
\begin{align}
    U_\mathrm{HE}  = e^{i\hat{H}}.
\end{align}
When using quantum phase estimation, this Hamiltonian oracle must be normalized so that the eigenvalue spectrum lies between 0 and $2\pi$ to avoid aliasing, specially relevant in the limit of arbitrary input states. This is achieved by using a modified Hamiltonian, $H' = c_1H + c_2\mathbb{1}$, where $c_1$ and $c_2$ are normalization constants. The choice of these constants depends on the Hamiltonian type and the input state used in the algorithm. The constant $c_1$ is crucial as it rescales the eigenvalue spectrum, influencing the required precision for the output phase. 
\newline 
\newline 
For general bounded Hermitian operators without symmetries, the ideal normalization constants are $c_1 = \frac{2\pi}{\Delta}$ and $c_2 = -\frac{2\pi}{\Delta}E_\mathrm{min}$, where $\Delta = E_\mathrm{max} - E_\mathrm{min}$ is the spectral range, and $E_\mathrm{min}/E_\mathrm{max}$ are the minimum/maximum eigenvalues of the Hamiltonian operator, $\hat{H}$. These normalization constants are suitable for arbitrary input states. Although these eigenvalues are not known in advance, it is possible to use approximate eigenvalues $\tilde{E}_\mathrm{min}/\tilde{E}_\mathrm{max}$ satisfying the conditions, $E_\mathrm{min} \leq \tilde{E}_\mathrm{min}$ and $E_\mathrm{max} \leq \tilde{E}_\mathrm{max}$. As an example, the methods presented in Appendix A and Section IV may be used to provide numerically tractable estimates of $\tilde{E}_\mathrm{min}/\tilde{E}_\mathrm{max}$.
\newline 
\newline 
On the other hand, if the Hamiltonian preserves a specific symmetry transformation, different rescaling parameters should be used. Assuming the input state to the QPE algorithm belongs to the symmetry sector $\mu$, the normalization constants should ideally become equal to $c_1 = \frac{2\pi}{\Delta_\mu}$ and $c_2 = -\frac{2\pi}{\Delta_\mu}E^\mu_\mathrm{min}$, where $\Delta_\mu = E^\mu_\mathrm{max} - E^\mu_\mathrm{min}$, and $E^\mu_\mathrm{min}/E^\mu_\mathrm{max}$ are the minimum/maximum eigenvalues of the Hamiltonian operator $\hat{H}$ in symmetry sector $\mu$ ($\hat{H}_{\mu}$). The rescaling constants are suitable for arbitrary input states belonging to symmetry sector $\mu$. As above, the methods presented in Appendix A and Section IV may be used to provide numerically tractable estimates of $\tilde{E}^\mu_\mathrm{min}/\tilde{E}^\mu_\mathrm{max}$ belonging to a given symmetry sector $\mu$. 
\newline
\newline
\emph{Block-Encoding Oracle.} The second oracle we consider is the block-encoding oracle model defined by,
\begin{align}
    U_{\mathrm{BE}} = 
    \begin{pmatrix}
        \hat{H}/\lambda & * \\ 
        * & *
    \end{pmatrix}.
\end{align}
Block-encoding oracles have similar applications to HE oracles since these two input models are in fact equivalent if we ignore the overhead that comes from converting between them. Within the context of ground-state energy estimation, block-encoding is often used as a building block of the qubitization walk operator, defined as $\mathcal{W} = \exp\left(\pm i\cos^{-1}(\hat{H}/\lambda) \right)$. Since the output phase, $\varphi = \pm\cos^{-1}(E/\lambda)$, is limited between $0$ and $\pi$, aliasing is not a fundamental concern in this context. The primary consideration is ensuring that $\mathcal{W}$ remains unitary for any input state. This implies the Hamiltonian oracle must have an eigenvalue spectrum bounded between -1 and 1.
\newline 
\newline 
For general bounded Hermitian operators without any symmetries, the normalization constant $\lambda$ is equal to the 1-norm in Eqn \eqref{LCU}. Since the 1-norm is highly dependent on the specifics of the Hamiltonian LCU expansion, including orbital optimization and tensor compression methods \cite{ koridon2021orbital, von2021quantum,lee2021even}, it is possible to provide a rough lower bound using the spectral range, $\frac{\Delta}{2} \leq \|\hat{H}\|$, as highlighted in Theorem 1 \cite{loaiza2023reducing}. This implies that the rescaling for general Hamiltonians in the block-encoding oracle model is also limited by the spectral range for arbitrary input states.
\newline 
\newline 
The situation changes when Hamiltonian preserves set of symmetries $\mathbb{S}=\{\hat{S}_m\}$. As highlighted in \cite{loaiza2023reducing,loaiza2023block}, it is possible to define two new Hamiltonians, $H_S$, and $H_{BI}$, which themselves have different eigenvalue spectra from $\hat{H}$.

\begin{definition}[Symmetry-shifted Hamiltonian]
The symmetry-shifted Hamiltonian is defined as, 
\begin{align}
    \hat{H}_\mathrm{S} = \hat{H} - \hat{f}(\mathbb{S}),
\end{align}
where $\hat{f}(\mathbb{S})$ is a well-behaved function of the Hamiltonian's symmetries, obeying $[\hat{H},\hat{f}(\mathbb{S})]=0$. In a given symmetry sector $\mu$, the symmetry function $\hat{f}(\mathbb{S})$ will obey the eigenvalue relation, $\hat{f}(\mathbb{S})|\psi_\mu\rangle=f_\mu|\psi_\mu\rangle$.
\end{definition}
\begin{definition}[Block-invariant Hamiltonian]
The block-invariant Hamiltonian is defined as, 
\begin{align}
    \hat{H}_\mathrm{BI} = \hat{H}_\mathrm{S} -\hat{B}(\hat{f}(\mathbb{S})-f_\mu) ,
    \label{block_invariant_Hamiltonian}
\end{align}
where $\hat{B}$ is a Hermitian operator that obeys the same symmetries as the Hamiltonian, $[\hat{B},\hat{f}(\mathbb{S})]=0$, but it is assumed that it does not generally commute with the Hamiltonian. The operator $\hat{f}(\mathbb{S})$ is another well-behaved function of the Hamiltonian's symmetries, $[\hat{H},\hat{f}(\mathbb{S})]=0$, obeying the eigenvalue relation, $\hat{f}(\mathbb{S})|\psi_\mu\rangle=f_\mu|\psi_\mu\rangle$. 
\end{definition}

While the symmetry-shifted Hamiltonian $\hat{H}_S$ preserves the eigenvalue spectrum for all possible symmetry-preserving input states apart from symmetry-sector-dependent identity shifts, the block-invariant Hamiltonian $\hat{H}_\mathrm{BI}$ preserves the correct eigenvalue spectrum for the symmetry sector $\mu$ only. As a result, it is possible to show that the symmetry sector bound acts as the true lower bound, $ \frac{\Delta_\mu}{2} \leq \lambda_\mathrm{BI}$, for arbitrary input states. The  details of this inequality are presented explicitly in the section below.

\section{Symmetry-aware spectral bounds} 

To take into account the effect of symmetries, we use the results from the previous section and explicitly write the Hamiltonian $\hat{H}$ in the block-diagonal basis, $\hat{H}=\bigoplus_\mu\hat{H}_\mu,$ where $\hat{H}_\mu$ is a sub-block defined with respect to a particular symmetry sector $\mu$. This leads to the following corollary, which describes the hierarchy of inequalities:

\begin{restatable}[Symmetry-aware spectral bounds]{corollary}{SymmetryBound}
\label{thm:SymmetryBound}
For a bounded block-diagonal Hermitian operator, $\hat{H}=\bigoplus_\mu \hat{H}_\mu$, the total spectral bound ($\Delta/2$) will obey the following hierarchy of lower bounds,
\begin{align}
    \tfrac{1}{2}\Delta_{\mu} \leq \tfrac{1}{2}\Delta_{\mathrm{s}} \leq \tfrac{1}{2}\Delta \, ,
\end{align}
where $\Delta_{\mu} = E_\mathrm{max}^\mu - E_\mathrm{min}^\mu$ represents the spectral range within a specific symmetry sector $\mu$, $\Delta_{\mathrm{s}} = \max_\mu (E_\mathrm{max}^\mu - E_\mathrm{min}^\mu)$ denotes the largest spectral range across all sectors, and $\Delta = \max_\mu (E_\mathrm{max}^\mu) - \min_\mu (E^\mu_\mathrm{min})$ encapsulates the full spectral range of $\hat{H}$. Here, $E^\mu_\mathrm{max}$($E^\mu_\mathrm{min}$) is the highest(lowest) eigenvalue within the symmetry sector $\mu$.
\end{restatable}
\begin{proof}
    Note that the original spectral bound can be expressed as:
    \begin{equation}
        \frac{\Delta}{2} = \frac{1}{2}\left( \max_\mu E_\mathrm{max}^\mu - \min_\mu E^\mu_\mathrm{min} \right).
    \end{equation}
    Based on this observation, the hierarchy of inequalities follows from the properties of max and min functions. For completeness, we demonstrate how this works by first writing, 
    \begin{align}
         \Delta &=  \max_\mu E_\mathrm{max}^\mu - \min_\mu E^\mu_\mathrm{min} \\
                &=  \max_\mu E_\mathrm{max}^\mu - \theta \\
                &= \max_\mu (E_\mathrm{max}^\mu - \theta )                 
    \end{align}
     where $\theta = \min_\mu E_{\mathrm{min}}^\mu$. By definition, $\theta \leq E^\mu_\mathrm{min}$ for a particular symmetry sector $\mu$. As a result, we obtain the following inequality,
    \begin{align}
        \Delta_s \equiv \max_\mu (E_\mathrm{max}^\mu - E_\mathrm{min}^\mu) \leq \Delta.
    \end{align}
    The last inequality, $(E_\mathrm{max}^\mu - E_\mathrm{min}^\mu) \leq \max_\mu (E_\mathrm{max}^\mu - E_\mathrm{min}^\mu)$, naturally follows. This creates the desired hierarchy of bounds and therefore completes the proof. 
\end{proof}
While this Corollary proves that a hierarchy of spectral bounds exists, it does not prove that a symmetry-shifted $\hat{H}_\mathrm{S}$ or block-invariant $\hat{H}_\mathrm{BI}$ Hamiltonian must have corresponding 1-norms that are restricted by either of these two bounds. The following two theorems prove that these lower bounds apply directly to symmetry-shifted and block-invariant Hamiltonians introduced in refs. \cite{loaiza2023reducing,loaiza2023block}.

\begin{restatable}[Symmetry-shifted 1-norm bound]{theorem}{SymmetrySpectralBound}
\label{thm:SymmetrySpectralBound}
For a bounded symmetry-shifted Hermitian operator, $\hat H_\mathrm{S}$, all of its possible LCU decompositions will have an associated 1-norm $\lambda_\mathrm{S}$ which is lower bounded by the spectral symmetry bound, $\Delta_s/2 \equiv \tfrac{1}{2}\max_\mu (E_\mathrm{max}^\mu - E_\mathrm{min}^\mu)$, 
\begin{align}
    \tfrac{1}{2}\Delta_s \leq \lambda_\mathrm{S}.
\end{align}
\end{restatable}

\begin{proof}
    Assuming that $E_\mathrm{max}^{\mu}/E_\mathrm{min}^\mu$ are the extremal eigenvalues of $\hat{H}$ within each symmetry sector $\mu$, the corresponding extremal eigenvalues of $\hat{H}_\mathrm{S}$ will be $(E_\mathrm{max}^{\mu}-f_\mu)/(E_\mathrm{min}^\mu-f_\mu)$ respectively. Using Theorem 1, the regular spectral bound for $\hat{H}_\mathrm{S}$ may be written as:
    \begin{align}
        \frac{1}{2}\left( \max_\mu (E_\mathrm{max}^\mu-f_\mu) - \min_\mu (E^\mu_\mathrm{min} - f_\mu) \right) \leq \lambda_\mathrm{S}.
    \end{align}
    We now wish to show that the quantity on the left hand side is strictly lower bounded by $\Delta_s/2$. This is obtained by writing,
    \begin{align}
        \frac{1}{2}\left( \max_\mu (E_\mathrm{max}^\mu-f_\mu) - \min_\mu (E^\mu_\mathrm{min} - f_\mu) \right) &= \frac{1}{2}\left( \max_\mu (E_\mathrm{max}^\mu-f_\mu) + \max_\mu ( f_\mu - E^\mu_\mathrm{min}) \right) \\
        &\geq \frac{1}{2}\max_\mu \left(  E_\mathrm{max}^\mu-f_\mu + f_\mu - E^\mu_\mathrm{min} \right). 
    \end{align}
    In the first line, we used the relation between max and min functions, $\max(-G)=-\min(G)$, and in the second line we used the triangle inequality, $\max(G_1+G_2)\leq \max(G_1)+\max(G_2)$. Since the quantity in the last line is exactly equal to $\Delta_s/2$, the proof is complete.
\end{proof}

\begin{restatable}[Block-invariant 1-norm bound]{theorem}{SymmetrySectorBound}
\label{thm:SymmetrySectorBound}
Consider the block-invariant Hermitian operator, $\hat H_\mathrm{BI}$, defined in \eqref{block_invariant_Hamiltonian} with an optimal symmetry-shifted Hamiltonian $\hat{H}_\mathrm{S}$. All of its possible LCU decompositions will have an associated 1-norm $\lambda_\mathrm{BI}$ which is lower bounded by the symmetry sector bound, $\Delta_\mu/2 \equiv \tfrac{1}{2}(E_\mathrm{max}^\mu - E_\mathrm{min}^\mu)$, 
\begin{align}
    \tfrac{1}{2}\Delta_\mu \leq \lambda_\mathrm{BI}.
\end{align}
\end{restatable}
\begin{proof}
    The first term in \eqref{block_invariant_Hamiltonian} corresponds to the symmetry-shifted Hamiltonian. If we assume that the symmetry-shifted Hamiltonian has been optimized so that it saturates its bound, $\tfrac{1}{2}\Delta_s$, then it is possible to show that $\lambda_\mathrm{BI}$ can become smaller than $\tfrac{1}{2}\Delta_s$ for a certain set of operators $\hat{B}$ and $\hat{f}(S)$. This illustrates that $\tfrac{1}{2}\Delta_s$ is not a proper lower bound. Nevertheless, Corollary 2 shows that $\lambda_{\mathrm{BI}}$ must remain lower bounded by the symmetry sector bound, $\tfrac{1}{2}\Delta_\mu$, regardless of the chosen operators $\hat{B}$ and $\hat{f}(\mathbb{S})$. This shows that $\tfrac{1}{2}\Delta_\mu$ is the true lower bound of the 1-norm of the block-invariant Hamiltonian, thereby completing the proof.
\end{proof}

Ultimately, while the conventional spectral bound, $\Delta/2$, provides a lower bound on the 1-norm $\lambda$ of the original Hamiltonian, Theorem 3 shows that $\Delta_s/2$ is the lower bound of the 1-norm of symmetry-shifted Hamiltonian proposed in \cite{loaiza2023reducing}. Furthermore, Theorem 4 proves that $\Delta_{\mu}/2$ lower bounds the so-called \emph{block-invariant} Hamiltonian proposed in \cite{loaiza2023block}. This substantiates the prevailing intuition that the symmetry-sector spectral range should inherently determine the query complexity lower bound for Hamiltonian oracles in quantum phase estimation. Given that all three Hamiltonians share an equivalent eigenvalue spectrum within the symmetry sector $\mu$, they are equally viable for constructing Hamiltonian oracles tailored to problem instances confined to this sector. These results not only set lower bounds on the query complexity for such oracles but also provide a hard cut-off in the expected improvement of tensor compression strategies designed to minimize the 1-norm. 

\subsection{Electronic Structure Hamiltonian}
Transitioning to the realm of electronic structure theory, we now apply these insights to the second quantized Hamiltonian, 
\begin{equation}
    \hat{H} = \sum_{\substack{pq \\\sigma}} h_{pq} \hat{a}^\dagger_{p,\sigma} \hat{a}_{q,\sigma} + \frac{1}{2}\sum_{\substack{pqrs \\\sigma\tau}} g_{pqrs} \hat{a}^\dagger_{p\sigma} \hat{a}^\dagger_{r\tau} \hat{a}_{s\tau} \hat{a}_{q\sigma},
    \label{electronic_hamiltonian}
\end{equation}
where $\{\sigma,\tau\}$ are spin indices, $\{p,q,r,s\}$ are spatial orbital indices, and $h_{pq}/g_{pqrs}$ are the conventional 1-electron/2-electron integrals defined with respect to the molecular spatial orbitals. The electronic structure Hamiltonian will generally conserve the total particle number $\hat{N}$, spin-projection $\hat{S}_z$, total spin $\hat{S}^2$ symmetries written as
\begin{align}
    \hat{N}   &= \sum_p ( a^\dagger_{p,\alpha} a_{p,\alpha} + a^\dagger_{p,\beta} a_{p,\beta} ) \, , \\
    \hat{S}_z &= \tfrac{1}{2}\sum_p ( a^\dagger_{p,\alpha} a_{p,\alpha} - a^\dagger_{p,\beta} a_{p,\beta} ) \, , \\
    \hat{S}^2 &= \hat{S}_+\hat{S}_- + \hat{S}_z(\hat{S}_z - 1) \, ,
\end{align}
where $\hat{S}_+ = \sum_p \hat{a}^\dagger_{p\alpha}\hat{a}_{p\beta}$ and $\hat{S}_- = (\hat{S}_+)^\dagger$. This Hamiltonian also remains invariant to any transformation involving a well-behaved function of these symmetries, $f(\hat{N},\hat{S}_z,\hat{S}^2)$. It is important to note that while the Numerical Results Section of this manuscript only considers particle number and spin symmetries, specific problem instances might also possess additional symmetries that should be included. For instance, permutational and point-group symmetries for specific molecular structures, such as rotation, reflection, and inversion, should also be considered. In addition, periodic symmetries in the context of crystalline systems could also greatly reduce the 1-norm \cite{ivanov2023quantum, rubin2023fault}. As a result, these bounds suggest that the extensive scaling behavior of block-invariant 1-norms should significantly differ from what has been observed in previous numerical studies.


\emph{Scaling analysis.} While symmetry-aware spectral bounds provide insights into the performance limits of quantum algorithms based on Hamiltonian oracle models, they do not provide any insight into the extensive scaling behavior with respect to system size. Previous numerical analysis of the scaling behavior of the 1-norm for the electronic structure Hamiltonian has been performed based on sparse, double factorization, and THC \cite{lee2021even,koridon2021orbital} representations. It has been shown to scale polynomially, $\lambda = \mathcal{O}(\text{poly}(N_\mathrm{orb}))$, defined with respect to the number of orbitals, $N_\mathrm{orb}$. The symmetry-aware spectral bounds presented in this manuscript suggest a different scaling that should incorporate some dependence with respect to the total particle number $\eta$ or other symmetry-sector eigenvalue. In Appendix A, we present theoretical results for the scaling behavior of the symmetry-aware spectral bounds for the electronic structure Hamiltonian, highlighted by the upper bound dependence, $\Delta_\mu \leq \mathcal{O}(g(N_\mathrm{orb})\eta^2)$. Our results find that an underlying spatial orbital dependence originates from the function $g(N_\mathrm{orb})$, which is heavily dependent on the type of basis set. In the numerical results section below, we show that the family of correlation-consistent basis sets demonstrate a sublinear dependence, $g(N_\mathrm{orb}) = \mathcal{O}(N_\mathrm{orb}^x)$ with $x\leq 1$, for different chemical systems. The following section outlines a numerically tractable procedure for computing these bounds for large systems and presents quantitative results for several benchmark cases.

\section{Numerical evaluation of symmetry-aware spectral bounds}

To numerically determine the three spectral bounds ($\Delta_\mu/2,\Delta_\mathrm{s}/2,$ and $\Delta/2$) for molecular systems, the smallest and largest eigenvalues, denoted by \(E^\mu_\mathrm{min}\) and \(E^\mu_\mathrm{max}\), must be computed in every symmetry sector \(\mu\). In the context of electronic structure theory, we consider the three symmetries, $(\hat{N},\hat{S}^2,\hat{S}_z)$, outlined above with each sector uniquely defined by the triplet, $\mu = (\eta, S, m_s)$, where $\eta$ equals the total number of electrons, while $S$ and $m_s$ represent the total spin and its projection, respectively. In the absence of external magnetic fields and spin-orbit coupling, as considered throughout this manuscript, the quantum numbers $\eta$ and $S$ are sufficient to uniquely identify all energy eigenvalues. While an ideal calculation would provide these eigenvalues with full configuration interaction (FCI)-level precision, the computational demand becomes prohibitive for large system sizes. To address this, we propose a numerically tractable method that involves performing orbital optimization to identify the determinant that minimizes or maximizes the energy within each sector, similar to a restricted Hartree-Fock level calculation suitable in both closed-shell and open-shell cases. It is important to emphasize that this method does not provide multi-reference capabilities and, therefore, may not accurately describe certain physical states, such as open-shell singlets (detailed in Appendix B). Despite this limitation, our method adheres to the variational principle, ensuring that the estimated spectral ranges for each sector remains a lower bound to the exact FCI solution (see Figure \ref{fig:variational_methods}),
\begin{equation}
     \Delta_\mu^{(\mathrm{HF})} \leq \Delta_\mu^{(\mathrm{FCI})}.
\end{equation}
This guarantees that our numerical estimates serve as a conservative lower bound to the actual FCI spectral range.  While more expensive methodologies (e.g., CISD, CCSD, or DMRG) could be used to provide more accurate assessments of the spectral bounds, we found that the orbital-optimized methodology remains useful in providing non-trivial estimates of the spectral range as well as its scaling behavior (see Appendix B for details). It is worth noting that for large system sizes, the calculation of the spectral bound ($\Delta/2$) and symmetry bound ($\Delta_s/2$) becomes computationally intensive due to the total number of symmetry sector. In a fixed basis set with $N_\mathrm{orb}$ spatial orbitals, the total number of electrons will be limited to the range \( [0, 2N_\mathrm{orb}] \) across the entire Hilbert space. In the absence of external magnetic fields or spin-orbit coupling, the latter resulting from relativistic adjustments to the Hamiltonian, the energy eigenvalues characterized by constant $\eta$ and $S$, but varying $m_s$, will be degenerate. This implies that the total number of symmetry sectors $(\eta,S)$ with unique energy eigenvalues will be given by,
\begin{equation}
    \frac{(N_\mathrm{orb}+1)(N_\mathrm{orb}+2)}{2}.
\end{equation}
While this number is only quadratic, it renders the numerical evaluation of $\Delta/2$ and $\Delta_s/2$ quite resource-intensive for large basis sets.
\begin{figure}[t!]
    \centering
    \includegraphics[width=.5\linewidth]{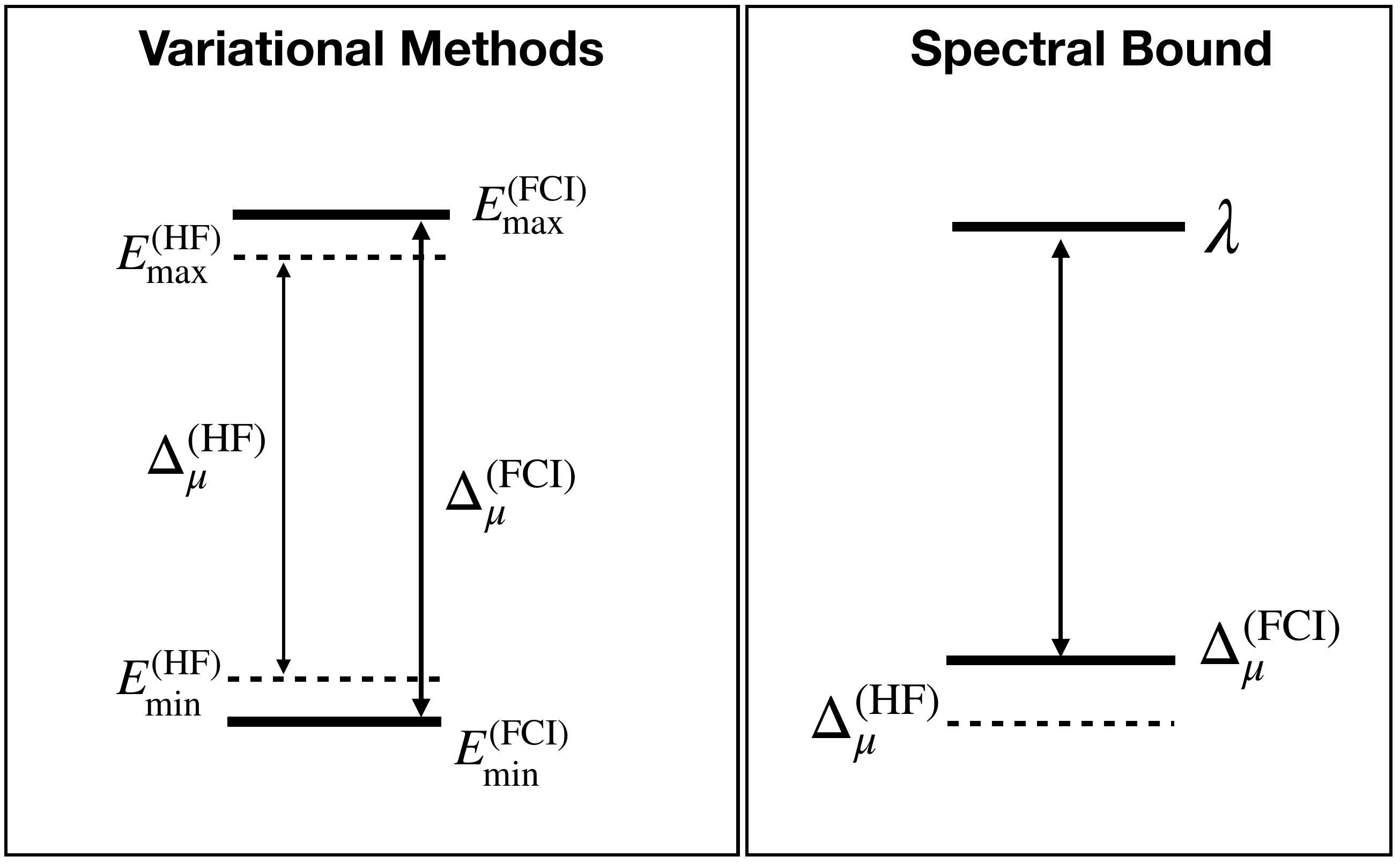}%
    \caption{ Overview of variational methods (left panel) and implications on the spectral bound numerical estimates (right panel).}
    \label{fig:variational_methods}
\end{figure}
\subsection{Orbital Optimization Procedure}
In this section, we outline an orbital optimization process that applies the same set of orbital parameters to both spin-up ($\alpha$) and spin-down ($\beta$) orbitals, and is applicable in both closed-shell and open-shell configurations. For convenience, we use ($\eta_\alpha,\eta_\beta$) notation to identify each symmetry sector, as opposed $(\eta,S)$ as above, where $\eta_\alpha$ and $\eta_\beta$ represent the number of spin-up and spin-down electrons. The proposed methodology finds the minimum eigenvalue by solving the optimization problem, $E_\mathrm{min}^{(\eta_\alpha,\eta_\beta)} = \min_{\boldsymbol{\kappa}} E_\mathrm{HF}(\boldsymbol{\kappa})$, where
\begin{align}
    E_\mathrm{HF}(\boldsymbol{\kappa})\! =\! \sum_{i}^{\eta_\alpha}h_{ii}(\boldsymbol{\kappa}) + \sum_{i}^{\eta_\beta} h_{ii}(\boldsymbol{\kappa})+ \! \frac{1}{2}\!\!\left( \!\sum_{i,j}^{\eta_\alpha,\eta_\alpha}\!\!g_{iijj}(\boldsymbol{\kappa}) + \!\sum_{i,j}^{\eta_\beta,\eta_\beta}\!\!g_{iijj}(\boldsymbol{\kappa})
    + 2\!\!\sum_{i,j}^{\eta_\alpha,\eta_\beta}\!\!g_{iijj}(\boldsymbol{\kappa})
    - \!\sum_{i,j}^{\eta_\alpha,\eta_\alpha}\!\!g_{ijji}(\boldsymbol{\kappa}) -
\!\sum_{i,j}^{\eta_\beta,\eta_\beta}\!\!g_{ijji}(\boldsymbol{\kappa}) \right)\!\!,
\end{align}
with orbital-transformed one-electron $h_{ij}$ and two-electron $g_{ijkl}$ tensor coefficients,
\begin{align}
    h_{ij}(\boldsymbol{\kappa}) &= \sum_{pq}h_{pq} [\exp(-\boldsymbol{\kappa})]_{pi}[\exp(-\boldsymbol{\kappa})]_{qj}, \\
    g_{ijkl}(\boldsymbol{\kappa}) &= \sum_{pqrs}g_{pqrs}[\exp(-\boldsymbol{\kappa})]_{pi}[\exp(-\boldsymbol{\kappa})]_{qj}
    [\exp(-\boldsymbol{\kappa})]_{rk}[\exp(-\boldsymbol{\kappa})]_{sl}
    ,
\end{align}
defined with respect to the unitary matrix, $\exp(-\boldsymbol{\kappa})$, written in terms of the anti-Hermitian matrix, $\boldsymbol{\kappa}$, where $\boldsymbol{\kappa}^\dagger = -\boldsymbol{\kappa}$. To improve the stability of the optimization procedure (especially relevant for finding the maximun eigenvalue), we used the pre-orthogonalized one-electron $h_{pq}$ and two-electron integrals $g_{pqrs}$,
\begin{align}
    h_{pq} &= \sum_{\mu\nu}h_{\mu\nu}[S^{-\frac{1}{2}}]_{\mu p} [S^{-\frac{1}{2}}]_{\nu q},  \label{hij}\\
    g_{pqrs} &= \sum_{\mu\nu\rho\sigma} g_{\mu\nu\rho\sigma}
    [S^{-\frac{1}{2}}]_{\mu p} [S^{-\frac{1}{2}}]_{\nu q}
    [S^{-\frac{1}{2}}]_{\rho r} [S^{-\frac{1}{2}}]_{\sigma s},
    \label{gijkl}
\end{align}
as our initial starting guess. These quantities are defined with respect to the atomic orbital overlap matrix ($S_{\mu \nu}$) as well as the one-electron ($h_{\mu\nu}$) and two-electron ($g_{\mu\nu\rho\sigma}$) integrals defined in the atomic orbital basis. For the optimization of $N_\mathrm{orb}$ real spatial orbitals, it is sufficient to consider the real part of $\boldsymbol{\kappa}$, which contains only $N_\mathrm{orb}(N_\mathrm{orb}-1)/2$ parameters. The maximum eigenvalue, $E_\mathrm{max}^{(\eta_\alpha,\eta_\beta)}$, can be found with the same numerical minimizer by taking the negative of the one-electron and two-electron integrals. It is important to note that for RHF and ROHF techniques, the associated wave function will be an eigenfunction of the total spin operator, $\hat{S}^2$, with an expectation value equal to,
\begin{align}
    \braket{\hat{S}^2} = \left( \frac{\eta_\alpha-\eta_\beta}{2} \right)\!\!\left( \frac{\eta_\alpha-\eta_\beta}{2} + 1\right),
\end{align}
assuming $\eta_\alpha\geq \eta_\beta$. As a result, the spin quantum number $s$ is equal to the spin-projection quantum number $|m_s|$. Although this might appear restrictive, this implies that in scenarios devoid of magnetic fields and spin-orbit interactions, adjusting $\eta_\alpha$ and $\eta_\beta$ suffices to encompass all possible symmetry sectors characterized by unique energy eigenvalues.

\section{Numerical Results}  

Maintaining the tradition of previous work \cite{reiher2017elucidating, lee2021even}, we first compute the symmetry-aware spectral bounds in Table 1 for FeMoco based on two distinct active space models. The approach by Reiher \emph{et al.} \cite{reiher2017elucidating} employs a $(54e,54o)$ active space, targeting the singlet ground-state energy ($S=0$). The model by Li \emph{et al.}  \cite{li2019electronic} captures the open-shell nature of the cluster with a $(113e,76o)$ active space, specifying a total spin of $S=3/2$. Additionally, we determine the spectral bounds for the largest active space of P450, as defined by Goings \emph{et al.} \cite{goings2022reliably}, aiming for the ground-state energy within a $(63e,58o)$ active space and a total spin of $S=5/2$. Finally, we consider the heme-artemisinin system relevant to drug design proposed by Cortes \emph{et al.} \cite{cortes2023fault}, requiring a $(90e,83o)$ active space with a total spin of $S=1$.

\begin{table}[b!]
\centering
\begin{tabular}{|c|l|a|a|b|a|}
\hline
System   & Bound & 1-body & 2-body & Total (incoh.)  & Total   \\ \hline
\multirow{3}{*}{\textbf{FeMoco}} & $\Delta/2$ & 38.6 & 111.7 & 150.3 & 135.0 \\
 &  $\Delta_\mathrm{s}/2$    & 19.5 & 34.0 & 53.5 & 45.7\\
(Reiher) &  $\Delta_\mu/2$ & 19.5 & 26.4 & 45.9 & 38.8 \\ \hline \hline    
\multirow{3}{*}{\textbf{FeMoco}} & $\Delta/2$      & 478.1 & 261.2  & 739.3 & 559.4 \\
                    & $\Delta_\mathrm{s}/2$    & 58.5  & 77.3   & 135.8 & 106.2 \\
                   (Li) & $\Delta_\mu/2$ & 41.2  & 67.7   & 108.9 & 57.6\\ \hline \hline
\multirow{3}{*}{\textbf{P450} } & $\Delta/2$      & 68.0 & 164.9 & 232.9 & 209.7 \\
                        & $\Delta_\mathrm{s}/2$    & 32.2 & 43.0  & 75.2  & 57.5 \\
                        & $\Delta_\mu/2$ & 31.8 & 37.5  & 69.3  & 49.9 \\ \hline \hline    
\multirow{3}{*}{\textbf{Heme-} } & $\Delta/2$      & 168.6 & 233.6 & 402.2   & 361.6 \\
                        & $\Delta_\mathrm{s}/2$    & 107.3 & 86.5  & 193.8  & 167.5 \\
   \textbf{Artemisinin} & $\Delta_\mu/2$           & 106.3 & 81.7  & 188.0  & 135.3 \\ \hline    
\end{tabular}
\caption{The spectral bound $\Delta/2$, the symmetry bound $\Delta_s/2$ and the symmetry-sector bound $\Delta_\mu/2$  for active space benchmark systems proposed in Refs.~\cite{reiher2017elucidating, lee2021even, goings2022reliably, cortes2023fault}, reported in units of Hartree. The spectral bounds of the 1-body and 2-body terms of the Hamiltonian are given separately, and their sum is the incoherent total spectral bound. The spectral bound of the total Hamiltonian in the last column is always lower.}
\label{tab:results_1}
\end{table}

\begin{figure*}[ht!]
    \centering
    \begin{minipage}{0.33\textwidth}
        \centering
        \includegraphics[width=\linewidth]{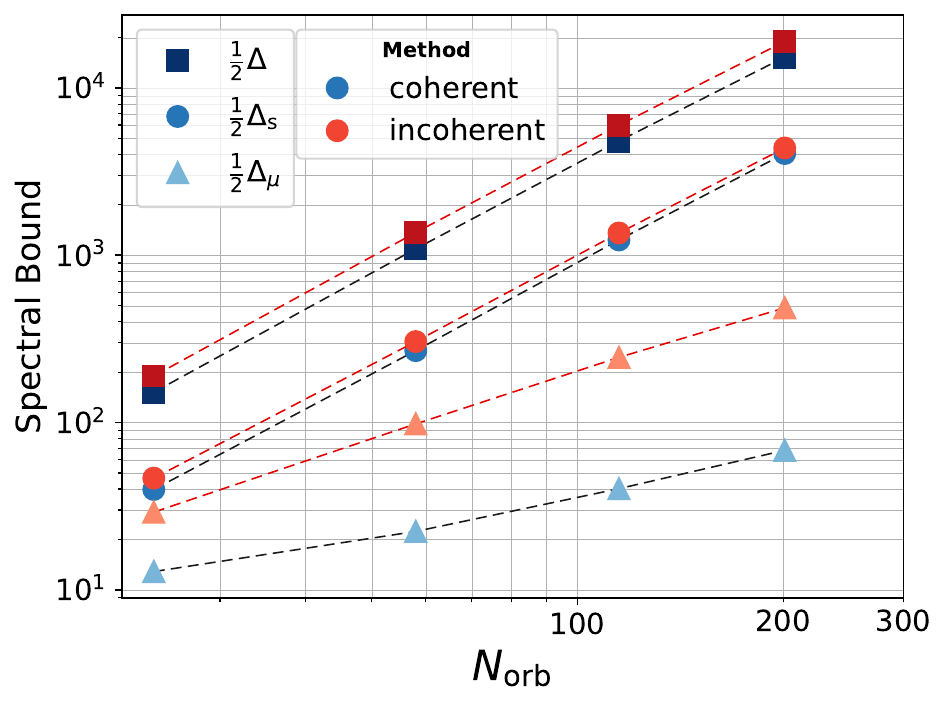} 
    \makebox[\linewidth][c]{\hspace{0.5cm}BeH$_2$}
    \end{minipage}\hfill
    \begin{minipage}{0.33\textwidth}
        \centering
        \includegraphics[width=\linewidth]{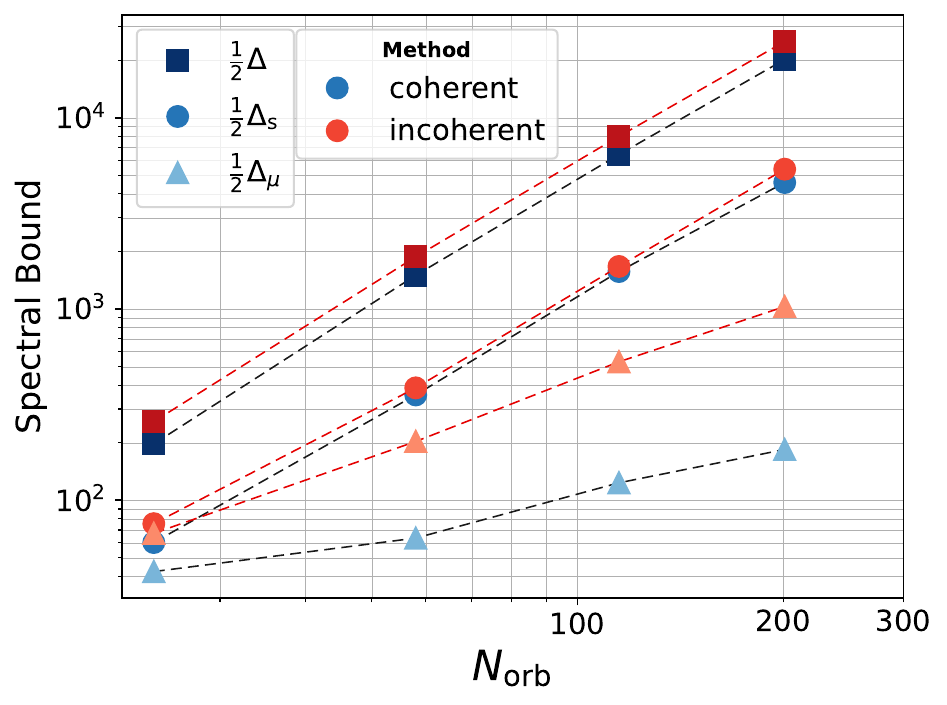} 
    \makebox[\linewidth][c]{\hspace{0.5cm}H$_2$O}
    \end{minipage}\hfill
    \begin{minipage}{0.33\textwidth}
        \centering
        \includegraphics[width=\linewidth]{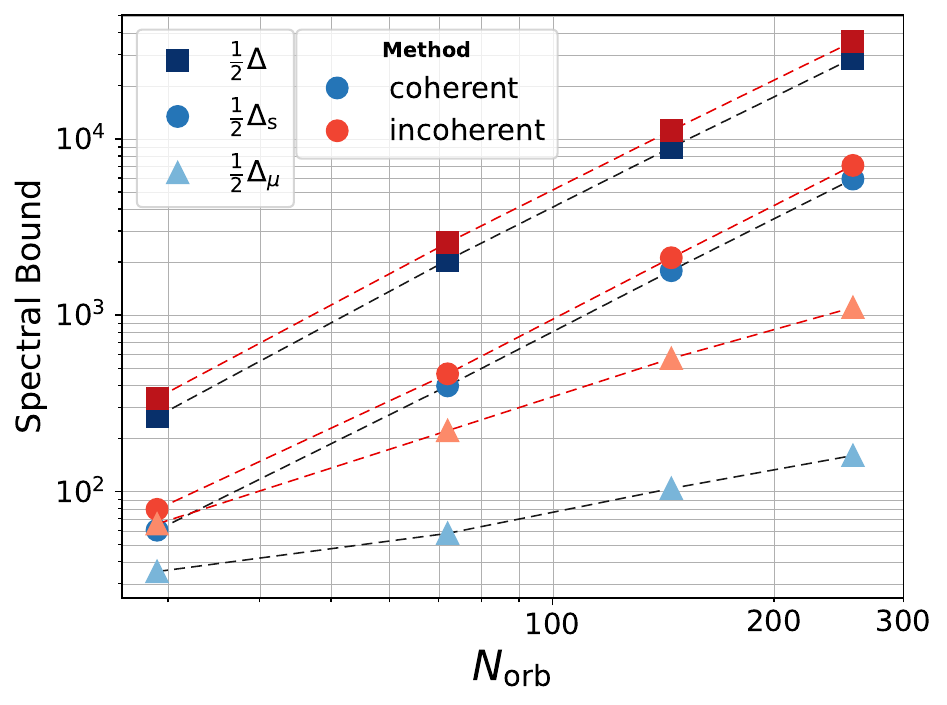} 
    \makebox[\linewidth][c]{\hspace{0.5cm}NH$_3$}
    \end{minipage}
\caption{The spectral bound $\Delta/2$, symmetry bound $\Delta_s/2$ and symmetry-sector bound $\Delta_\mu/2$ as a function of spatial orbitals, $N_\mathrm{orb}$, with correlation-consistent basis sets
(cc-pVDZ, cc-pVTZ, cc-pVQZ, cc-pV5Z) for three molecules BeH$_2$, H$_2$O, and NH$_3$, with all values reported in units of Hartree.}
\label{fig:fsp-basis}    
\end{figure*}
In addition to providing the three spectral bounds ($\Delta_\mu/2,\Delta_\mathrm{s}/2,$ and $\Delta/2$) for the total Hamiltonian in the last column, we have also computed the spectral bound for the 1-body and 2-body terms separately (the details of the decomposition is outlined in Appendix A). The second last column (in gray) denotes the sum of the 1-body and 2-body bounds, which we call the incoherent spectral bound. This column should serve as the reference point for most contemporary tensor factorization techniques, including double factorization and tensor hypercontraction (THC), due to the fact that these methods factorize the 1-body and 2-body terms independently  \cite{von2021quantum,lee2021even}. 

As highlighted in Table \ref{tab:results_1}, the discrepancy between the incoherent and coherent spectral bounds suggests that applying tensor compression methods to a \emph{dressed} four-index tensor, which integrates both 1-body and 2-body contributions, might provide a further reduction in the 1-norm. Our prior research demonstrated the efficacy of this concept within the context of observable estimation \cite{cortes2023fault}, achieving up to a four-fold reduction in the 1-norm for specific cases. This decrease in the 1-norm results from the natural cancellation of energy contributions within the dressed tensor. Nonetheless, a thorough examination of this approach is warranted, particularly in the context of estimating ground-state energies, with careful consideration of its impact on $C_{\mathcal{W}[H]}$. 

Traditional tensor compression methods, like double factorization and THC, which do not incorporate symmetry information, are constrained by the spectral bound ($\Delta/2$) found in the top row. As a point of reference, the published 1-norm values for FeMoco based on THC correspond to 306 (Reiher \emph{et al.}~\cite{reiher2017elucidating}) and 1202 (Lee \emph{et al.}~\cite{lee2021even}) and for the P450 Hamiltonian
of 388.9 (Goings \emph{et al.}~\cite{goings2022reliably}), in units of Hartree. These numbers suggest that additional enhancements to THC, such as applying regularization techniques \cite{goings2022reliably, oumarou2022accelerating}, refining initial starting guesses, or relaxing the constraints on the THC rank, could at most achieve improvement factors of around 2.0 and 1.6, respectively. On the other hand, the symmetry-based bounds suggest that tensor compression techniques that exploit symmetry information in the future might be able to achieve an approximate improvement factor of up to 6 and 21 for the FeMoco-based benchmark systems outlined above. A recently introduced symmetry-shifted version of double factorization (SCDF) was shown to violate the conventional bound ($\Delta/2$) and achieved a 1-norm of 78.0 for FeMoco (Reiher) and 111.3 for P450 (Rocca \emph{et al.}~\cite{rocca2024reducing}), but remained the symmetry-shifted spectral bound ($\Delta_\mathrm{s}/2$) as expected.

\emph{Complete Basis Set Limit Scaling.} To further highlight the differences between these bounds, we compute the scaling behavior of the three spectral bounds ($\Delta_\mu/2,\Delta_\mathrm{s}/2,$ and $\Delta/2$) for BeH$_2$, H$_2$O, and NH$_3$. These molecules and their respective geometries were selected based on prior research \cite{loaiza2023reducing, loaiza2023block} because they are small enough to enable complete basis set extrapolation analysis. The numerical analysis is conducted as a function of basis set size within the correlation-consistent (cc-pVnZ) basis sets, ranging from cc-pVDZ to cc-pV5Z. The size of each basis set is quantified explicitly by the number of spatial orbitals, denoted as $N_\mathrm{orb}$. For each of these systems, the symmetry sector bound ($\Delta_\mu/2$) is fixed to the singlet ($S=0$) sector with the total electron counts of $6$, $10$, and $10$, respectively. The resulting data are illustrated in Fig.~\ref{fig:fsp-basis}, where both the incoherent (orange markers) and regular (blue markers) spectral bounds are presented for comparative analysis.

\begin{figure}[b!]
    \includegraphics[width=7.4cm]{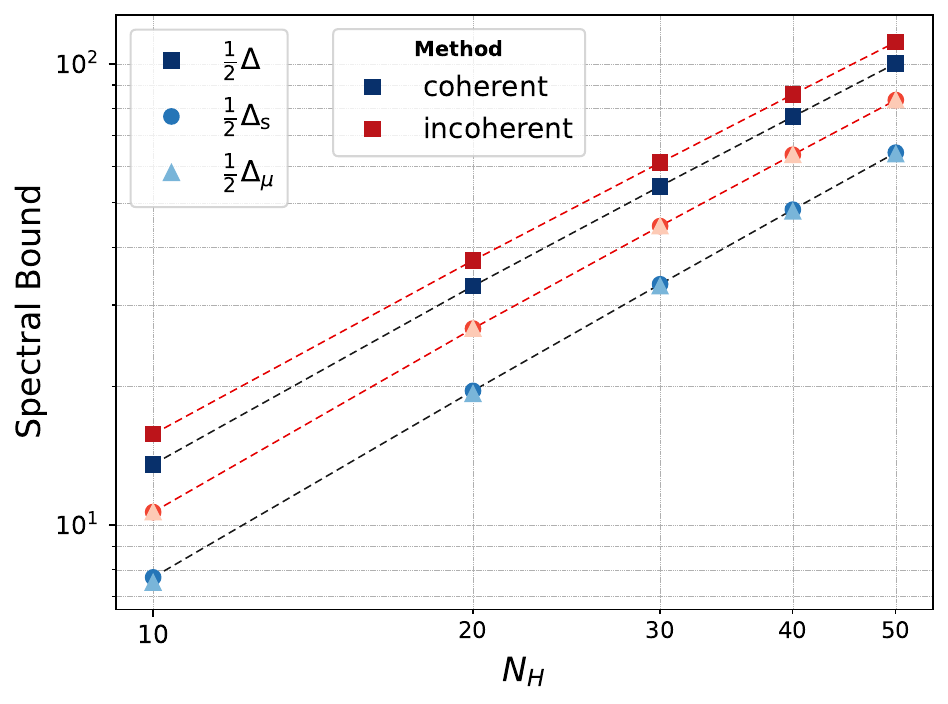}
    \caption{Symmetry-aware spectral bounds as a function of the number of hydrogen atoms, $N_H$, reported in units of Hartree.}
    \label{fig:fsp-hydrogen}
\end{figure}

The coherent spectral bound is lower than its incoherent counterpart in every instance. This disparity becomes increasingly pronounced with larger basis sets, exhibiting nearly an order of magnitude difference between the symmetry-sector bounds. Consistent with the theoretical findings from the previous section, our analysis reveals markedly different scaling behaviors between the regular bound, the symmetry bound, and the symmetry-sector bound, as evidenced by their distinct slopes. This leads to a nearly three-order-of-magnitude difference between the regular spectral bound ($\Delta/2$) and the symmetry-sector bound ($\Delta_\mu/2$) in the cc-pV5Z basis. In particular, the symmetry-sector bound $\Delta_\mu/2$ exhibits a sublinear slope, $g(N)=\mathcal{O}(N^x)$, with $x\leq 1$ across all three molecular systems.

It is important to differentiate this orbital number dependence in the context of larger basis sets from the continuum limit analysis of the $H_4$ system by Lee \emph{et al.} \cite{lee2021even}, which employed a fixed basis set (cc-pVTZ) and varied the number of orbitals within. That scenario could lead to an almost constant scaling, $g(N)=\mathcal{O}(1)$, for symmetry-aware spectral bounds given that the basis set parameters are held constant. The situation changes, however, with basis set extrapolation and the use of different basis sets, as in our study. Here, the sublinear dependence of $g(N)$ is attributed to the larger exponents used in cc-pV5Z compared to smaller basis sets such as cc-pVDZ. These larger exponents provide enhanced resolution of the Coulomb cusp, but also increase the magnitude of the integral coefficients. Larger basis sets also provide improved behavior at greater distances from the nucleus to achieve the complete basis set limit.

\emph{Thermodynamic Limit Scaling.} To conclude our numerical analysis, we also compute the spectral bounds of a linear hydrogen chain at half-filling as a function of the number of hydrogen atoms, $N_H$, ranging from 10 to 50. The base geometry consists of a linear array of hydrogen atoms separated by 0.74 \r{A}ngstrom expanded in the minimal STO-6G basis set. In Fig.~\ref{fig:fsp-hydrogen}, we find that while there is nearly an order of magnitude difference between the regular spectral bound ($\Delta/2$) and the symmetry bound ($\Delta_\mathrm{s}/2$), the associated difference between the symmetry bound and the symmetry sector bound ($\Delta_\mu/2$) is found to be nearly negligible. We attribute this to the half-filling sector typically having the largest spectral range among all symmetry sectors for this benchmark system. Since the gap between conventional and symmetry-aware bounds is more pronounced in the low-filling regime (as also highlighted in Fig.~\ref{fig:fsp-basis}), symmetry-shifted tensor compression techniques should be sufficient to saturate the achievable 1-norm bounds in the half-filling regime.

\section{Conclusion} 

This study introduces symmetry-aware spectral bounds relevant to a wide variety of quantum algorithms in physics and chemistry. For qubitization-based QPE, we find an asymptotic Toffoli complexity of $\mathcal{O}(\Delta_\mu N_\mathrm{orb} )$, assuming a linear $N_\mathrm{orb}$ scaling of the qubitization walk operator as observed in THC \cite{lee2021even}. Combined with the upper bound scaling of the symmetry-sector bound, $\Delta_\mu = \mathcal{\tilde{O}}(N_\mathrm{orb}^x\eta^2)$, derived in this work, we predict a final Toffoli complexity of $ \mathcal{\tilde{O}}(N_\mathrm{orb}^{1+x}\eta^2)$. Unlike previous work, this scaling has an explicit dependence on the number of particles $\eta$ making it much more favorable in the continuum and complete basis set limit. In contrast, first-quantization methods using plane waves have a Toffoli complexity of $\mathcal{O}(\eta^{8/3}N_\mathrm{orb}^{1/3})$ \cite{su2021fault}. This raises the question of whether block-encoding strategies with explicit knowledge of symmetries could be constructed. Such dependence would enhance the competitiveness of second-quantized methods using localized Gaussian orbitals compared with first-quantization techniques in these limits. 

Our findings also indicate that although modern tensor compression methods, especially THC, are close to approaching the conventional spectral bound $\Delta/2$, there remains a lot of room for improvement in reaching the symmetry ($\Delta_s/2$) and symmetry-sector ($\Delta_\mu/2)$ bounds. Loaiza \emph{et al.}~\cite{loaiza2023reducing, loaiza2023block} proposed symmetry-shifted and block-invariant Hamiltonians as a methodology for reducing the 1-norms. Recent work by Rocca \emph{et al.}~\cite{rocca2024reducing} exploited such symmetry shifts to reduce the 1-norm towards the symmetry bound $\Delta_s/2$ for FeMoco and the P450. 
However, there is still significant opportunity for advancements in tensor compression techniques, in particular by leveraging the full set of symmetries to reach 1-norms saturating the symmetry-sector bound $\Delta_\mu/2$. While our study has focused on performance limits tied to 1-norm, exploring other constraints, such as those related to the block-encoding cost, $C_{\mathcal{W}[H]}$, of the Hamiltonian remains critical. Future research should seek to define the optimal trade-off, or Pareto frontier, between these factors for enhancing quantum algorithm efficiency.
This exploration highlights multiple promising paths for future enhancements in quantum algorithm capabilities.

\bibliography{manuscript.bib}

\begin{thebibliography}{46}%
\makeatletter
\providecommand \@ifxundefined [1]{%
 \@ifx{#1\undefined}
}%
\providecommand \@ifnum [1]{%
 \ifnum #1\expandafter \@firstoftwo
 \else \expandafter \@secondoftwo
 \fi
}%
\providecommand \@ifx [1]{%
 \ifx #1\expandafter \@firstoftwo
 \else \expandafter \@secondoftwo
 \fi
}%
\providecommand \natexlab [1]{#1}%
\providecommand \enquote  [1]{``#1''}%
\providecommand \bibnamefont  [1]{#1}%
\providecommand \bibfnamefont [1]{#1}%
\providecommand \citenamefont [1]{#1}%
\providecommand \href@noop [0]{\@secondoftwo}%
\providecommand \href [0]{\begingroup \@sanitize@url \@href}%
\providecommand \@href[1]{\@@startlink{#1}\@@href}%
\providecommand \@@href[1]{\endgroup#1\@@endlink}%
\providecommand \@sanitize@url [0]{\catcode `\\12\catcode `\$12\catcode `\&12\catcode `\#12\catcode `\^12\catcode `\_12\catcode `\%12\relax}%
\providecommand \@@startlink[1]{}%
\providecommand \@@endlink[0]{}%
\providecommand \url  [0]{\begingroup\@sanitize@url \@url }%
\providecommand \@url [1]{\endgroup\@href {#1}{\urlprefix }}%
\providecommand \urlprefix  [0]{URL }%
\providecommand \Eprint [0]{\href }%
\providecommand \doibase [0]{https://doi.org/}%
\providecommand \selectlanguage [0]{\@gobble}%
\providecommand \bibinfo  [0]{\@secondoftwo}%
\providecommand \bibfield  [0]{\@secondoftwo}%
\providecommand \translation [1]{[#1]}%
\providecommand \BibitemOpen [0]{}%
\providecommand \bibitemStop [0]{}%
\providecommand \bibitemNoStop [0]{.\EOS\space}%
\providecommand \EOS [0]{\spacefactor3000\relax}%
\providecommand \BibitemShut  [1]{\csname bibitem#1\endcsname}%
\let\auto@bib@innerbib\@empty
\bibitem [{\citenamefont {Low}\ and\ \citenamefont {Chuang}(2017)}]{low2017optimal}%
  \BibitemOpen
  \bibfield  {author} {\bibinfo {author} {\bibfnamefont {G.~H.}\ \bibnamefont {Low}}\ and\ \bibinfo {author} {\bibfnamefont {I.~L.}\ \bibnamefont {Chuang}},\ }\href@noop {} {\bibfield  {journal} {\bibinfo  {journal} {Physical review letters}\ }\textbf {\bibinfo {volume} {118}},\ \bibinfo {pages} {010501} (\bibinfo {year} {2017})}\BibitemShut {NoStop}%
\bibitem [{\citenamefont {Low}\ and\ \citenamefont {Chuang}(2019)}]{low2019hamiltonian}%
  \BibitemOpen
  \bibfield  {author} {\bibinfo {author} {\bibfnamefont {G.~H.}\ \bibnamefont {Low}}\ and\ \bibinfo {author} {\bibfnamefont {I.~L.}\ \bibnamefont {Chuang}},\ }\href@noop {} {\bibfield  {journal} {\bibinfo  {journal} {Quantum}\ }\textbf {\bibinfo {volume} {3}},\ \bibinfo {pages} {163} (\bibinfo {year} {2019})}\BibitemShut {NoStop}%
\bibitem [{\citenamefont {Gily{\'e}n}\ \emph {et~al.}(2019)\citenamefont {Gily{\'e}n}, \citenamefont {Su}, \citenamefont {Low},\ and\ \citenamefont {Wiebe}}]{gilyen2019quantum}%
  \BibitemOpen
  \bibfield  {author} {\bibinfo {author} {\bibfnamefont {A.}~\bibnamefont {Gily{\'e}n}}, \bibinfo {author} {\bibfnamefont {Y.}~\bibnamefont {Su}}, \bibinfo {author} {\bibfnamefont {G.~H.}\ \bibnamefont {Low}},\ and\ \bibinfo {author} {\bibfnamefont {N.}~\bibnamefont {Wiebe}},\ }in\ \href@noop {} {\emph {\bibinfo {booktitle} {Proceedings of the 51st Annual ACM SIGACT Symposium on Theory of Computing}}}\ (\bibinfo {year} {2019})\ pp.\ \bibinfo {pages} {193--204}\BibitemShut {NoStop}%
\bibitem [{\citenamefont {Chan}\ \emph {et~al.}(2023)\citenamefont {Chan}, \citenamefont {Mu{\~n}oz-Ramo},\ and\ \citenamefont {Fitzpatrick}}]{chan2023simulating}%
  \BibitemOpen
  \bibfield  {author} {\bibinfo {author} {\bibfnamefont {H.~H.~S.}\ \bibnamefont {Chan}}, \bibinfo {author} {\bibfnamefont {D.}~\bibnamefont {Mu{\~n}oz-Ramo}},\ and\ \bibinfo {author} {\bibfnamefont {N.}~\bibnamefont {Fitzpatrick}},\ }\href@noop {} {\bibfield  {journal} {\bibinfo  {journal} {arXiv preprint arXiv:2303.06161}\ } (\bibinfo {year} {2023})}\BibitemShut {NoStop}%
\bibitem [{\citenamefont {Rall}(2020)}]{rall2020quantum}%
  \BibitemOpen
  \bibfield  {author} {\bibinfo {author} {\bibfnamefont {P.}~\bibnamefont {Rall}},\ }\href@noop {} {\bibfield  {journal} {\bibinfo  {journal} {Physical Review A}\ }\textbf {\bibinfo {volume} {102}},\ \bibinfo {pages} {022408} (\bibinfo {year} {2020})}\BibitemShut {NoStop}%
\bibitem [{\citenamefont {Tong}\ \emph {et~al.}(2021)\citenamefont {Tong}, \citenamefont {An}, \citenamefont {Wiebe},\ and\ \citenamefont {Lin}}]{tong2021fast}%
  \BibitemOpen
  \bibfield  {author} {\bibinfo {author} {\bibfnamefont {Y.}~\bibnamefont {Tong}}, \bibinfo {author} {\bibfnamefont {D.}~\bibnamefont {An}}, \bibinfo {author} {\bibfnamefont {N.}~\bibnamefont {Wiebe}},\ and\ \bibinfo {author} {\bibfnamefont {L.}~\bibnamefont {Lin}},\ }\href@noop {} {\bibfield  {journal} {\bibinfo  {journal} {Physical Review A}\ }\textbf {\bibinfo {volume} {104}},\ \bibinfo {pages} {032422} (\bibinfo {year} {2021})}\BibitemShut {NoStop}%
\bibitem [{\citenamefont {Steudtner}\ \emph {et~al.}(2023)\citenamefont {Steudtner}, \citenamefont {Morley-Short}, \citenamefont {Pol}, \citenamefont {Sim}, \citenamefont {Cortes}, \citenamefont {Loipersberger}, \citenamefont {Parrish}, \citenamefont {Degroote}, \citenamefont {Moll}, \citenamefont {Santagati},\ and\ \citenamefont {Streif}}]{steudtner2023fault}%
  \BibitemOpen
  \bibfield  {author} {\bibinfo {author} {\bibfnamefont {M.}~\bibnamefont {Steudtner}}, \bibinfo {author} {\bibfnamefont {S.}~\bibnamefont {Morley-Short}}, \bibinfo {author} {\bibfnamefont {W.}~\bibnamefont {Pol}}, \bibinfo {author} {\bibfnamefont {S.}~\bibnamefont {Sim}}, \bibinfo {author} {\bibfnamefont {C.~L.}\ \bibnamefont {Cortes}}, \bibinfo {author} {\bibfnamefont {M.}~\bibnamefont {Loipersberger}}, \bibinfo {author} {\bibfnamefont {R.~M.}\ \bibnamefont {Parrish}}, \bibinfo {author} {\bibfnamefont {M.}~\bibnamefont {Degroote}}, \bibinfo {author} {\bibfnamefont {N.}~\bibnamefont {Moll}}, \bibinfo {author} {\bibfnamefont {R.}~\bibnamefont {Santagati}},\ and\ \bibinfo {author} {\bibfnamefont {M.}~\bibnamefont {Streif}},\ }\href {https://doi.org/10.22331/q-2023-11-06-1164} {\bibfield  {journal} {\bibinfo  {journal} {{Quantum}}\ }\textbf {\bibinfo {volume} {7}},\ \bibinfo {pages} {1164} (\bibinfo {year} {2023})}\BibitemShut {NoStop}%
\bibitem [{\citenamefont {An}\ \emph {et~al.}(2023)\citenamefont {An}, \citenamefont {Childs},\ and\ \citenamefont {Lin}}]{an2023quantum}%
  \BibitemOpen
  \bibfield  {author} {\bibinfo {author} {\bibfnamefont {D.}~\bibnamefont {An}}, \bibinfo {author} {\bibfnamefont {A.~M.}\ \bibnamefont {Childs}},\ and\ \bibinfo {author} {\bibfnamefont {L.}~\bibnamefont {Lin}},\ }\href@noop {} {\bibfield  {journal} {\bibinfo  {journal} {arXiv preprint arXiv:2312.03916}\ } (\bibinfo {year} {2023})}\BibitemShut {NoStop}%
\bibitem [{\citenamefont {Krovi}(2023)}]{krovi2023improved}%
  \BibitemOpen
  \bibfield  {author} {\bibinfo {author} {\bibfnamefont {H.}~\bibnamefont {Krovi}},\ }\href@noop {} {\bibfield  {journal} {\bibinfo  {journal} {Quantum}\ }\textbf {\bibinfo {volume} {7}},\ \bibinfo {pages} {913} (\bibinfo {year} {2023})}\BibitemShut {NoStop}%
\bibitem [{\citenamefont {Santagati}\ \emph {et~al.}(2024)\citenamefont {Santagati}, \citenamefont {Aspuru-Guzik}, \citenamefont {Babbush}, \citenamefont {Degroote}, \citenamefont {Gonz{\'a}lez}, \citenamefont {Kyoseva}, \citenamefont {Moll}, \citenamefont {Oppel}, \citenamefont {Parrish}, \citenamefont {Rubin}, \citenamefont {Streif}, \citenamefont {Tautermann}, \citenamefont {Weiss}, \citenamefont {Wiebe},\ and\ \citenamefont {Utschig-Utschig}}]{Santagati2024}%
  \BibitemOpen
  \bibfield  {author} {\bibinfo {author} {\bibfnamefont {R.}~\bibnamefont {Santagati}}, \bibinfo {author} {\bibfnamefont {A.}~\bibnamefont {Aspuru-Guzik}}, \bibinfo {author} {\bibfnamefont {R.}~\bibnamefont {Babbush}}, \bibinfo {author} {\bibfnamefont {M.}~\bibnamefont {Degroote}}, \bibinfo {author} {\bibfnamefont {L.}~\bibnamefont {Gonz{\'a}lez}}, \bibinfo {author} {\bibfnamefont {E.}~\bibnamefont {Kyoseva}}, \bibinfo {author} {\bibfnamefont {N.}~\bibnamefont {Moll}}, \bibinfo {author} {\bibfnamefont {M.}~\bibnamefont {Oppel}}, \bibinfo {author} {\bibfnamefont {R.~M.}\ \bibnamefont {Parrish}}, \bibinfo {author} {\bibfnamefont {N.~C.}\ \bibnamefont {Rubin}}, \bibinfo {author} {\bibfnamefont {M.}~\bibnamefont {Streif}}, \bibinfo {author} {\bibfnamefont {C.~S.}\ \bibnamefont {Tautermann}}, \bibinfo {author} {\bibfnamefont {H.}~\bibnamefont {Weiss}}, \bibinfo {author} {\bibfnamefont {N.}~\bibnamefont {Wiebe}},\ and\ \bibinfo {author} {\bibfnamefont {C.}~\bibnamefont {Utschig-Utschig}},\ }\href@noop {} {\bibfield
   {journal} {\bibinfo  {journal} {Nature Physics}\ } (\bibinfo {year} {2024})}\BibitemShut {NoStop}%
\bibitem [{\citenamefont {Martyn}\ \emph {et~al.}(2021)\citenamefont {Martyn}, \citenamefont {Rossi}, \citenamefont {Tan},\ and\ \citenamefont {Chuang}}]{martyn2021grand}%
  \BibitemOpen
  \bibfield  {author} {\bibinfo {author} {\bibfnamefont {J.~M.}\ \bibnamefont {Martyn}}, \bibinfo {author} {\bibfnamefont {Z.~M.}\ \bibnamefont {Rossi}}, \bibinfo {author} {\bibfnamefont {A.~K.}\ \bibnamefont {Tan}},\ and\ \bibinfo {author} {\bibfnamefont {I.~L.}\ \bibnamefont {Chuang}},\ }\href@noop {} {\bibfield  {journal} {\bibinfo  {journal} {PRX Quantum}\ }\textbf {\bibinfo {volume} {2}},\ \bibinfo {pages} {040203} (\bibinfo {year} {2021})}\BibitemShut {NoStop}%
\bibitem [{\citenamefont {Dutkiewicz}\ \emph {et~al.}(2022)\citenamefont {Dutkiewicz}, \citenamefont {Terhal},\ and\ \citenamefont {O'Brien}}]{dutkiewicz2022heisenberg}%
  \BibitemOpen
  \bibfield  {author} {\bibinfo {author} {\bibfnamefont {A.}~\bibnamefont {Dutkiewicz}}, \bibinfo {author} {\bibfnamefont {B.~M.}\ \bibnamefont {Terhal}},\ and\ \bibinfo {author} {\bibfnamefont {T.~E.}\ \bibnamefont {O'Brien}},\ }\href@noop {} {\bibfield  {journal} {\bibinfo  {journal} {Quantum}\ }\textbf {\bibinfo {volume} {6}},\ \bibinfo {pages} {830} (\bibinfo {year} {2022})}\BibitemShut {NoStop}%
\bibitem [{\citenamefont {Lin}\ and\ \citenamefont {Tong}(2022)}]{lin2022heisenberg}%
  \BibitemOpen
  \bibfield  {author} {\bibinfo {author} {\bibfnamefont {L.}~\bibnamefont {Lin}}\ and\ \bibinfo {author} {\bibfnamefont {Y.}~\bibnamefont {Tong}},\ }\href@noop {} {\bibfield  {journal} {\bibinfo  {journal} {PRX Quantum}\ }\textbf {\bibinfo {volume} {3}},\ \bibinfo {pages} {010318} (\bibinfo {year} {2022})}\BibitemShut {NoStop}%
\bibitem [{\citenamefont {Dong}\ \emph {et~al.}(2022)\citenamefont {Dong}, \citenamefont {Lin},\ and\ \citenamefont {Tong}}]{dong2022ground}%
  \BibitemOpen
  \bibfield  {author} {\bibinfo {author} {\bibfnamefont {Y.}~\bibnamefont {Dong}}, \bibinfo {author} {\bibfnamefont {L.}~\bibnamefont {Lin}},\ and\ \bibinfo {author} {\bibfnamefont {Y.}~\bibnamefont {Tong}},\ }\href@noop {} {\bibfield  {journal} {\bibinfo  {journal} {PRX Quantum}\ }\textbf {\bibinfo {volume} {3}},\ \bibinfo {pages} {040305} (\bibinfo {year} {2022})}\BibitemShut {NoStop}%
\bibitem [{\citenamefont {Li}\ \emph {et~al.}(2023)\citenamefont {Li}, \citenamefont {Ni},\ and\ \citenamefont {Ying}}]{li2023adaptive}%
  \BibitemOpen
  \bibfield  {author} {\bibinfo {author} {\bibfnamefont {H.}~\bibnamefont {Li}}, \bibinfo {author} {\bibfnamefont {H.}~\bibnamefont {Ni}},\ and\ \bibinfo {author} {\bibfnamefont {L.}~\bibnamefont {Ying}},\ }\href@noop {} {\bibfield  {journal} {\bibinfo  {journal} {Physical Review A}\ }\textbf {\bibinfo {volume} {108}},\ \bibinfo {pages} {062408} (\bibinfo {year} {2023})}\BibitemShut {NoStop}%
\bibitem [{\citenamefont {Ding}\ and\ \citenamefont {Lin}(2023)}]{ding2023even}%
  \BibitemOpen
  \bibfield  {author} {\bibinfo {author} {\bibfnamefont {Z.}~\bibnamefont {Ding}}\ and\ \bibinfo {author} {\bibfnamefont {L.}~\bibnamefont {Lin}},\ }\href@noop {} {\bibfield  {journal} {\bibinfo  {journal} {PRX Quantum}\ }\textbf {\bibinfo {volume} {4}},\ \bibinfo {pages} {020331} (\bibinfo {year} {2023})}\BibitemShut {NoStop}%
\bibitem [{\citenamefont {Rendon}\ \emph {et~al.}(2022)\citenamefont {Rendon}, \citenamefont {Izubuchi},\ and\ \citenamefont {Kikuchi}}]{rendon2022effects}%
  \BibitemOpen
  \bibfield  {author} {\bibinfo {author} {\bibfnamefont {G.}~\bibnamefont {Rendon}}, \bibinfo {author} {\bibfnamefont {T.}~\bibnamefont {Izubuchi}},\ and\ \bibinfo {author} {\bibfnamefont {Y.}~\bibnamefont {Kikuchi}},\ }\href@noop {} {\bibfield  {journal} {\bibinfo  {journal} {Physical Review D}\ }\textbf {\bibinfo {volume} {106}},\ \bibinfo {pages} {034503} (\bibinfo {year} {2022})}\BibitemShut {NoStop}%
\bibitem [{\citenamefont {Xiong}\ \emph {et~al.}(2022)\citenamefont {Xiong}, \citenamefont {Ng}, \citenamefont {Long},\ and\ \citenamefont {Hanzo}}]{xiong2022dual}%
  \BibitemOpen
  \bibfield  {author} {\bibinfo {author} {\bibfnamefont {Y.}~\bibnamefont {Xiong}}, \bibinfo {author} {\bibfnamefont {S.~X.}\ \bibnamefont {Ng}}, \bibinfo {author} {\bibfnamefont {G.-L.}\ \bibnamefont {Long}},\ and\ \bibinfo {author} {\bibfnamefont {L.}~\bibnamefont {Hanzo}},\ }\href@noop {} {\bibfield  {journal} {\bibinfo  {journal} {IEEE Signal Processing Letters}\ }\textbf {\bibinfo {volume} {29}},\ \bibinfo {pages} {1222} (\bibinfo {year} {2022})}\BibitemShut {NoStop}%
\bibitem [{\citenamefont {Chakraborty}\ \emph {et~al.}(2019)\citenamefont {Chakraborty}, \citenamefont {Gily\'{e}n},\ and\ \citenamefont {Jeffery}}]{chakraborty2018power}%
  \BibitemOpen
  \bibfield  {author} {\bibinfo {author} {\bibfnamefont {S.}~\bibnamefont {Chakraborty}}, \bibinfo {author} {\bibfnamefont {A.}~\bibnamefont {Gily\'{e}n}},\ and\ \bibinfo {author} {\bibfnamefont {S.}~\bibnamefont {Jeffery}},\ }in\ \href {https://doi.org/10.4230/LIPIcs.ICALP.2019.33} {\emph {\bibinfo {booktitle} {46th International Colloquium on Automata, Languages, and Programming (ICALP 2019)}}},\ \bibinfo {series} {Leibniz International Proceedings in Informatics (LIPIcs)}, Vol.\ \bibinfo {volume} {132},\ \bibinfo {editor} {edited by\ \bibinfo {editor} {\bibfnamefont {C.}~\bibnamefont {Baier}}, \bibinfo {editor} {\bibfnamefont {I.}~\bibnamefont {Chatzigiannakis}}, \bibinfo {editor} {\bibfnamefont {P.}~\bibnamefont {Flocchini}},\ and\ \bibinfo {editor} {\bibfnamefont {S.}~\bibnamefont {Leonardi}}}\ (\bibinfo  {publisher} {Schloss Dagstuhl -- Leibniz-Zentrum f{\"u}r Informatik},\ \bibinfo {address} {Dagstuhl, Germany},\ \bibinfo {year} {2019})\ pp.\ \bibinfo {pages} {33:1--33:14}\BibitemShut {NoStop}%
\bibitem [{\citenamefont {Babbush}\ \emph {et~al.}(2018)\citenamefont {Babbush}, \citenamefont {Gidney}, \citenamefont {Berry}, \citenamefont {Wiebe}, \citenamefont {McClean}, \citenamefont {Paler}, \citenamefont {Fowler},\ and\ \citenamefont {Neven}}]{babbush2018encoding}%
  \BibitemOpen
  \bibfield  {author} {\bibinfo {author} {\bibfnamefont {R.}~\bibnamefont {Babbush}}, \bibinfo {author} {\bibfnamefont {C.}~\bibnamefont {Gidney}}, \bibinfo {author} {\bibfnamefont {D.~W.}\ \bibnamefont {Berry}}, \bibinfo {author} {\bibfnamefont {N.}~\bibnamefont {Wiebe}}, \bibinfo {author} {\bibfnamefont {J.}~\bibnamefont {McClean}}, \bibinfo {author} {\bibfnamefont {A.}~\bibnamefont {Paler}}, \bibinfo {author} {\bibfnamefont {A.}~\bibnamefont {Fowler}},\ and\ \bibinfo {author} {\bibfnamefont {H.}~\bibnamefont {Neven}},\ }\href@noop {} {\bibfield  {journal} {\bibinfo  {journal} {Physical Review X}\ }\textbf {\bibinfo {volume} {8}},\ \bibinfo {pages} {041015} (\bibinfo {year} {2018})}\BibitemShut {NoStop}%
\bibitem [{\citenamefont {von Burg}\ \emph {et~al.}(2021)\citenamefont {von Burg}, \citenamefont {Low}, \citenamefont {H{\"a}ner}, \citenamefont {Steiger}, \citenamefont {Reiher}, \citenamefont {Roetteler},\ and\ \citenamefont {Troyer}}]{von2021quantum}%
  \BibitemOpen
  \bibfield  {author} {\bibinfo {author} {\bibfnamefont {V.}~\bibnamefont {von Burg}}, \bibinfo {author} {\bibfnamefont {G.~H.}\ \bibnamefont {Low}}, \bibinfo {author} {\bibfnamefont {T.}~\bibnamefont {H{\"a}ner}}, \bibinfo {author} {\bibfnamefont {D.~S.}\ \bibnamefont {Steiger}}, \bibinfo {author} {\bibfnamefont {M.}~\bibnamefont {Reiher}}, \bibinfo {author} {\bibfnamefont {M.}~\bibnamefont {Roetteler}},\ and\ \bibinfo {author} {\bibfnamefont {M.}~\bibnamefont {Troyer}},\ }\href@noop {} {\bibfield  {journal} {\bibinfo  {journal} {Physical Review Research}\ }\textbf {\bibinfo {volume} {3}},\ \bibinfo {pages} {033055} (\bibinfo {year} {2021})}\BibitemShut {NoStop}%
\bibitem [{\citenamefont {Lee}\ \emph {et~al.}(2021)\citenamefont {Lee}, \citenamefont {Berry}, \citenamefont {Gidney}, \citenamefont {Huggins}, \citenamefont {McClean}, \citenamefont {Wiebe},\ and\ \citenamefont {Babbush}}]{lee2021even}%
  \BibitemOpen
  \bibfield  {author} {\bibinfo {author} {\bibfnamefont {J.}~\bibnamefont {Lee}}, \bibinfo {author} {\bibfnamefont {D.~W.}\ \bibnamefont {Berry}}, \bibinfo {author} {\bibfnamefont {C.}~\bibnamefont {Gidney}}, \bibinfo {author} {\bibfnamefont {W.~J.}\ \bibnamefont {Huggins}}, \bibinfo {author} {\bibfnamefont {J.~R.}\ \bibnamefont {McClean}}, \bibinfo {author} {\bibfnamefont {N.}~\bibnamefont {Wiebe}},\ and\ \bibinfo {author} {\bibfnamefont {R.}~\bibnamefont {Babbush}},\ }\href@noop {} {\bibfield  {journal} {\bibinfo  {journal} {PRX Quantum}\ }\textbf {\bibinfo {volume} {2}},\ \bibinfo {pages} {030305} (\bibinfo {year} {2021})}\BibitemShut {NoStop}%
\bibitem [{\citenamefont {Motta}\ \emph {et~al.}(2021)\citenamefont {Motta}, \citenamefont {Ye}, \citenamefont {McClean}, \citenamefont {Li}, \citenamefont {Minnich}, \citenamefont {Babbush},\ and\ \citenamefont {Chan}}]{motta2021low}%
  \BibitemOpen
  \bibfield  {author} {\bibinfo {author} {\bibfnamefont {M.}~\bibnamefont {Motta}}, \bibinfo {author} {\bibfnamefont {E.}~\bibnamefont {Ye}}, \bibinfo {author} {\bibfnamefont {J.~R.}\ \bibnamefont {McClean}}, \bibinfo {author} {\bibfnamefont {Z.}~\bibnamefont {Li}}, \bibinfo {author} {\bibfnamefont {A.~J.}\ \bibnamefont {Minnich}}, \bibinfo {author} {\bibfnamefont {R.}~\bibnamefont {Babbush}},\ and\ \bibinfo {author} {\bibfnamefont {G.~K.-L.}\ \bibnamefont {Chan}},\ }\href@noop {} {\bibfield  {journal} {\bibinfo  {journal} {npj Quantum Information}\ }\textbf {\bibinfo {volume} {7}},\ \bibinfo {pages} {83} (\bibinfo {year} {2021})}\BibitemShut {NoStop}%
\bibitem [{\citenamefont {Cohn}\ \emph {et~al.}(2021)\citenamefont {Cohn}, \citenamefont {Motta},\ and\ \citenamefont {Parrish}}]{cohn2021quantum}%
  \BibitemOpen
  \bibfield  {author} {\bibinfo {author} {\bibfnamefont {J.}~\bibnamefont {Cohn}}, \bibinfo {author} {\bibfnamefont {M.}~\bibnamefont {Motta}},\ and\ \bibinfo {author} {\bibfnamefont {R.~M.}\ \bibnamefont {Parrish}},\ }\href@noop {} {\bibfield  {journal} {\bibinfo  {journal} {PRX Quantum}\ }\textbf {\bibinfo {volume} {2}},\ \bibinfo {pages} {040352} (\bibinfo {year} {2021})}\BibitemShut {NoStop}%
\bibitem [{\citenamefont {Oumarou}\ \emph {et~al.}(2022)\citenamefont {Oumarou}, \citenamefont {Scheurer}, \citenamefont {Parrish}, \citenamefont {Hohenstein},\ and\ \citenamefont {Gogolin}}]{oumarou2022accelerating}%
  \BibitemOpen
  \bibfield  {author} {\bibinfo {author} {\bibfnamefont {O.}~\bibnamefont {Oumarou}}, \bibinfo {author} {\bibfnamefont {M.}~\bibnamefont {Scheurer}}, \bibinfo {author} {\bibfnamefont {R.~M.}\ \bibnamefont {Parrish}}, \bibinfo {author} {\bibfnamefont {E.~G.}\ \bibnamefont {Hohenstein}},\ and\ \bibinfo {author} {\bibfnamefont {C.}~\bibnamefont {Gogolin}},\ }\href@noop {} {\bibfield  {journal} {\bibinfo  {journal} {arXiv preprint arXiv:2212.07957}\ } (\bibinfo {year} {2022})}\BibitemShut {NoStop}%
\bibitem [{\citenamefont {Dunlap}\ \emph {et~al.}(1979)\citenamefont {Dunlap}, \citenamefont {Connolly},\ and\ \citenamefont {Sabin}}]{dunlap1979some}%
  \BibitemOpen
  \bibfield  {author} {\bibinfo {author} {\bibfnamefont {B.~I.}\ \bibnamefont {Dunlap}}, \bibinfo {author} {\bibfnamefont {J.}~\bibnamefont {Connolly}},\ and\ \bibinfo {author} {\bibfnamefont {J.}~\bibnamefont {Sabin}},\ }\href@noop {} {\bibfield  {journal} {\bibinfo  {journal} {The Journal of Chemical Physics}\ }\textbf {\bibinfo {volume} {71}},\ \bibinfo {pages} {3396} (\bibinfo {year} {1979})}\BibitemShut {NoStop}%
\bibitem [{\citenamefont {Werner}\ \emph {et~al.}(2003)\citenamefont {Werner}, \citenamefont {Manby},\ and\ \citenamefont {Knowles}}]{werner2003fast}%
  \BibitemOpen
  \bibfield  {author} {\bibinfo {author} {\bibfnamefont {H.-J.}\ \bibnamefont {Werner}}, \bibinfo {author} {\bibfnamefont {F.~R.}\ \bibnamefont {Manby}},\ and\ \bibinfo {author} {\bibfnamefont {P.~J.}\ \bibnamefont {Knowles}},\ }\href@noop {} {\bibfield  {journal} {\bibinfo  {journal} {The Journal of chemical physics}\ }\textbf {\bibinfo {volume} {118}},\ \bibinfo {pages} {8149} (\bibinfo {year} {2003})}\BibitemShut {NoStop}%
\bibitem [{\citenamefont {Pedersen}\ \emph {et~al.}(2009)\citenamefont {Pedersen}, \citenamefont {Aquilante},\ and\ \citenamefont {Lindh}}]{pedersen2009density}%
  \BibitemOpen
  \bibfield  {author} {\bibinfo {author} {\bibfnamefont {T.~B.}\ \bibnamefont {Pedersen}}, \bibinfo {author} {\bibfnamefont {F.}~\bibnamefont {Aquilante}},\ and\ \bibinfo {author} {\bibfnamefont {R.}~\bibnamefont {Lindh}},\ }\href@noop {} {\bibfield  {journal} {\bibinfo  {journal} {Theoretical Chemistry Accounts}\ }\textbf {\bibinfo {volume} {124}},\ \bibinfo {pages} {1} (\bibinfo {year} {2009})}\BibitemShut {NoStop}%
\bibitem [{\citenamefont {Hohenstein}\ and\ \citenamefont {Sherrill}(2010)}]{hohenstein2010density}%
  \BibitemOpen
  \bibfield  {author} {\bibinfo {author} {\bibfnamefont {E.~G.}\ \bibnamefont {Hohenstein}}\ and\ \bibinfo {author} {\bibfnamefont {C.~D.}\ \bibnamefont {Sherrill}},\ }\href@noop {} {\bibfield  {journal} {\bibinfo  {journal} {The Journal of chemical physics}\ }\textbf {\bibinfo {volume} {133}} (\bibinfo {year} {2010})}\BibitemShut {NoStop}%
\bibitem [{\citenamefont {Parrish}\ \emph {et~al.}(2013)\citenamefont {Parrish}, \citenamefont {Hohenstein}, \citenamefont {Schunck}, \citenamefont {Sherrill},\ and\ \citenamefont {Mart{\'\i}nez}}]{parrish2013exact}%
  \BibitemOpen
  \bibfield  {author} {\bibinfo {author} {\bibfnamefont {R.~M.}\ \bibnamefont {Parrish}}, \bibinfo {author} {\bibfnamefont {E.~G.}\ \bibnamefont {Hohenstein}}, \bibinfo {author} {\bibfnamefont {N.~F.}\ \bibnamefont {Schunck}}, \bibinfo {author} {\bibfnamefont {C.~D.}\ \bibnamefont {Sherrill}},\ and\ \bibinfo {author} {\bibfnamefont {T.~J.}\ \bibnamefont {Mart{\'\i}nez}},\ }\href@noop {} {\bibfield  {journal} {\bibinfo  {journal} {Physical Review Letters}\ }\textbf {\bibinfo {volume} {111}},\ \bibinfo {pages} {132505} (\bibinfo {year} {2013})}\BibitemShut {NoStop}%
\bibitem [{\citenamefont {Lee}\ \emph {et~al.}(2019)\citenamefont {Lee}, \citenamefont {Lin},\ and\ \citenamefont {Head-Gordon}}]{lee2019systematically}%
  \BibitemOpen
  \bibfield  {author} {\bibinfo {author} {\bibfnamefont {J.}~\bibnamefont {Lee}}, \bibinfo {author} {\bibfnamefont {L.}~\bibnamefont {Lin}},\ and\ \bibinfo {author} {\bibfnamefont {M.}~\bibnamefont {Head-Gordon}},\ }\href@noop {} {\bibfield  {journal} {\bibinfo  {journal} {Journal of chemical theory and computation}\ }\textbf {\bibinfo {volume} {16}},\ \bibinfo {pages} {243} (\bibinfo {year} {2019})}\BibitemShut {NoStop}%
\bibitem [{\citenamefont {Rubin}\ \emph {et~al.}(2022)\citenamefont {Rubin}, \citenamefont {Lee},\ and\ \citenamefont {Babbush}}]{rubin2022compressing}%
  \BibitemOpen
  \bibfield  {author} {\bibinfo {author} {\bibfnamefont {N.~C.}\ \bibnamefont {Rubin}}, \bibinfo {author} {\bibfnamefont {J.}~\bibnamefont {Lee}},\ and\ \bibinfo {author} {\bibfnamefont {R.}~\bibnamefont {Babbush}},\ }\href@noop {} {\bibfield  {journal} {\bibinfo  {journal} {Journal of Chemical Theory and Computation}\ }\textbf {\bibinfo {volume} {18}},\ \bibinfo {pages} {1480} (\bibinfo {year} {2022})}\BibitemShut {NoStop}%
\bibitem [{\citenamefont {Goings}\ \emph {et~al.}(2022)\citenamefont {Goings}, \citenamefont {White}, \citenamefont {Lee}, \citenamefont {Tautermann}, \citenamefont {Degroote}, \citenamefont {Gidney}, \citenamefont {Shiozaki}, \citenamefont {Babbush},\ and\ \citenamefont {Rubin}}]{goings2022reliably}%
  \BibitemOpen
  \bibfield  {author} {\bibinfo {author} {\bibfnamefont {J.~J.}\ \bibnamefont {Goings}}, \bibinfo {author} {\bibfnamefont {A.}~\bibnamefont {White}}, \bibinfo {author} {\bibfnamefont {J.}~\bibnamefont {Lee}}, \bibinfo {author} {\bibfnamefont {C.~S.}\ \bibnamefont {Tautermann}}, \bibinfo {author} {\bibfnamefont {M.}~\bibnamefont {Degroote}}, \bibinfo {author} {\bibfnamefont {C.}~\bibnamefont {Gidney}}, \bibinfo {author} {\bibfnamefont {T.}~\bibnamefont {Shiozaki}}, \bibinfo {author} {\bibfnamefont {R.}~\bibnamefont {Babbush}},\ and\ \bibinfo {author} {\bibfnamefont {N.~C.}\ \bibnamefont {Rubin}},\ }\href@noop {} {\bibfield  {journal} {\bibinfo  {journal} {Proceedings of the National Academy of Sciences}\ }\textbf {\bibinfo {volume} {119}},\ \bibinfo {pages} {e2203533119} (\bibinfo {year} {2022})}\BibitemShut {NoStop}%
\bibitem [{\citenamefont {Loaiza}\ \emph {et~al.}(2023)\citenamefont {Loaiza}, \citenamefont {Khah}, \citenamefont {Wiebe},\ and\ \citenamefont {Izmaylov}}]{loaiza2023reducing}%
  \BibitemOpen
  \bibfield  {author} {\bibinfo {author} {\bibfnamefont {I.}~\bibnamefont {Loaiza}}, \bibinfo {author} {\bibfnamefont {A.~M.}\ \bibnamefont {Khah}}, \bibinfo {author} {\bibfnamefont {N.}~\bibnamefont {Wiebe}},\ and\ \bibinfo {author} {\bibfnamefont {A.~F.}\ \bibnamefont {Izmaylov}},\ }\href@noop {} {\bibfield  {journal} {\bibinfo  {journal} {Quantum Science and Technology}\ }\textbf {\bibinfo {volume} {8}},\ \bibinfo {pages} {035019} (\bibinfo {year} {2023})}\BibitemShut {NoStop}%
\bibitem [{\citenamefont {Loaiza}\ and\ \citenamefont {Izmaylov}(2023)}]{loaiza2023block}%
  \BibitemOpen
  \bibfield  {author} {\bibinfo {author} {\bibfnamefont {I.}~\bibnamefont {Loaiza}}\ and\ \bibinfo {author} {\bibfnamefont {A.~F.}\ \bibnamefont {Izmaylov}},\ }\href@noop {} {\bibfield  {journal} {\bibinfo  {journal} {Journal of Chemical Theory and Computation}\ } (\bibinfo {year} {2023})}\BibitemShut {NoStop}%
\bibitem [{\citenamefont {Ivanov}\ \emph {et~al.}(2023)\citenamefont {Ivanov}, \citenamefont {S{\"u}nderhauf}, \citenamefont {Holzmann}, \citenamefont {Ellaby}, \citenamefont {Kerber}, \citenamefont {Jones},\ and\ \citenamefont {Camps}}]{ivanov2023quantum}%
  \BibitemOpen
  \bibfield  {author} {\bibinfo {author} {\bibfnamefont {A.~V.}\ \bibnamefont {Ivanov}}, \bibinfo {author} {\bibfnamefont {C.}~\bibnamefont {S{\"u}nderhauf}}, \bibinfo {author} {\bibfnamefont {N.}~\bibnamefont {Holzmann}}, \bibinfo {author} {\bibfnamefont {T.}~\bibnamefont {Ellaby}}, \bibinfo {author} {\bibfnamefont {R.~N.}\ \bibnamefont {Kerber}}, \bibinfo {author} {\bibfnamefont {G.}~\bibnamefont {Jones}},\ and\ \bibinfo {author} {\bibfnamefont {J.}~\bibnamefont {Camps}},\ }\href@noop {} {\bibfield  {journal} {\bibinfo  {journal} {Physical Review Research}\ }\textbf {\bibinfo {volume} {5}},\ \bibinfo {pages} {013200} (\bibinfo {year} {2023})}\BibitemShut {NoStop}%
\bibitem [{\citenamefont {Rubin}\ \emph {et~al.}(2023)\citenamefont {Rubin}, \citenamefont {Berry}, \citenamefont {Malone}, \citenamefont {White}, \citenamefont {Khattar}, \citenamefont {DePrince~III}, \citenamefont {Sicolo}, \citenamefont {K{\"u}ehn}, \citenamefont {Kaicher}, \citenamefont {Lee} \emph {et~al.}}]{rubin2023fault}%
  \BibitemOpen
  \bibfield  {author} {\bibinfo {author} {\bibfnamefont {N.~C.}\ \bibnamefont {Rubin}}, \bibinfo {author} {\bibfnamefont {D.~W.}\ \bibnamefont {Berry}}, \bibinfo {author} {\bibfnamefont {F.~D.}\ \bibnamefont {Malone}}, \bibinfo {author} {\bibfnamefont {A.~F.}\ \bibnamefont {White}}, \bibinfo {author} {\bibfnamefont {T.}~\bibnamefont {Khattar}}, \bibinfo {author} {\bibfnamefont {A.~E.}\ \bibnamefont {DePrince~III}}, \bibinfo {author} {\bibfnamefont {S.}~\bibnamefont {Sicolo}}, \bibinfo {author} {\bibfnamefont {M.}~\bibnamefont {K{\"u}ehn}}, \bibinfo {author} {\bibfnamefont {M.}~\bibnamefont {Kaicher}}, \bibinfo {author} {\bibfnamefont {J.}~\bibnamefont {Lee}}, \emph {et~al.},\ }\href@noop {} {\bibfield  {journal} {\bibinfo  {journal} {PRX Quantum}\ }\textbf {\bibinfo {volume} {4}},\ \bibinfo {pages} {040303} (\bibinfo {year} {2023})}\BibitemShut {NoStop}%
\bibitem [{\citenamefont {Koridon}\ \emph {et~al.}(2021)\citenamefont {Koridon}, \citenamefont {Yalouz}, \citenamefont {Senjean}, \citenamefont {Buda}, \citenamefont {O'Brien},\ and\ \citenamefont {Visscher}}]{koridon2021orbital}%
  \BibitemOpen
  \bibfield  {author} {\bibinfo {author} {\bibfnamefont {E.}~\bibnamefont {Koridon}}, \bibinfo {author} {\bibfnamefont {S.}~\bibnamefont {Yalouz}}, \bibinfo {author} {\bibfnamefont {B.}~\bibnamefont {Senjean}}, \bibinfo {author} {\bibfnamefont {F.}~\bibnamefont {Buda}}, \bibinfo {author} {\bibfnamefont {T.~E.}\ \bibnamefont {O'Brien}},\ and\ \bibinfo {author} {\bibfnamefont {L.}~\bibnamefont {Visscher}},\ }\href@noop {} {\bibfield  {journal} {\bibinfo  {journal} {Physical Review Research}\ }\textbf {\bibinfo {volume} {3}},\ \bibinfo {pages} {033127} (\bibinfo {year} {2021})}\BibitemShut {NoStop}%
\bibitem [{\citenamefont {Reiher}\ \emph {et~al.}(2017)\citenamefont {Reiher}, \citenamefont {Wiebe}, \citenamefont {Svore}, \citenamefont {Wecker},\ and\ \citenamefont {Troyer}}]{reiher2017elucidating}%
  \BibitemOpen
  \bibfield  {author} {\bibinfo {author} {\bibfnamefont {M.}~\bibnamefont {Reiher}}, \bibinfo {author} {\bibfnamefont {N.}~\bibnamefont {Wiebe}}, \bibinfo {author} {\bibfnamefont {K.~M.}\ \bibnamefont {Svore}}, \bibinfo {author} {\bibfnamefont {D.}~\bibnamefont {Wecker}},\ and\ \bibinfo {author} {\bibfnamefont {M.}~\bibnamefont {Troyer}},\ }\href@noop {} {\bibfield  {journal} {\bibinfo  {journal} {Proceedings of the national academy of sciences}\ }\textbf {\bibinfo {volume} {114}},\ \bibinfo {pages} {7555} (\bibinfo {year} {2017})}\BibitemShut {NoStop}%
\bibitem [{\citenamefont {Li}\ \emph {et~al.}(2019)\citenamefont {Li}, \citenamefont {Li}, \citenamefont {Dattani}, \citenamefont {Umrigar},\ and\ \citenamefont {Chan}}]{li2019electronic}%
  \BibitemOpen
  \bibfield  {author} {\bibinfo {author} {\bibfnamefont {Z.}~\bibnamefont {Li}}, \bibinfo {author} {\bibfnamefont {J.}~\bibnamefont {Li}}, \bibinfo {author} {\bibfnamefont {N.~S.}\ \bibnamefont {Dattani}}, \bibinfo {author} {\bibfnamefont {C.}~\bibnamefont {Umrigar}},\ and\ \bibinfo {author} {\bibfnamefont {G.~K.}\ \bibnamefont {Chan}},\ }\href@noop {} {\bibfield  {journal} {\bibinfo  {journal} {The Journal of chemical physics}\ }\textbf {\bibinfo {volume} {150}} (\bibinfo {year} {2019})}\BibitemShut {NoStop}%
\bibitem [{\citenamefont {Cortes}\ \emph {et~al.}(2024)\citenamefont {Cortes}, \citenamefont {Loipersberger}, \citenamefont {Parrish}, \citenamefont {Morley-Short}, \citenamefont {Pol}, \citenamefont {Sim}, \citenamefont {Steudtner}, \citenamefont {Tautermann}, \citenamefont {Degroote}, \citenamefont {Moll}, \citenamefont {Santagati},\ and\ \citenamefont {Streif}}]{cortes2023fault}%
  \BibitemOpen
  \bibfield  {author} {\bibinfo {author} {\bibfnamefont {C.~L.}\ \bibnamefont {Cortes}}, \bibinfo {author} {\bibfnamefont {M.}~\bibnamefont {Loipersberger}}, \bibinfo {author} {\bibfnamefont {R.~M.}\ \bibnamefont {Parrish}}, \bibinfo {author} {\bibfnamefont {S.}~\bibnamefont {Morley-Short}}, \bibinfo {author} {\bibfnamefont {W.}~\bibnamefont {Pol}}, \bibinfo {author} {\bibfnamefont {S.}~\bibnamefont {Sim}}, \bibinfo {author} {\bibfnamefont {M.}~\bibnamefont {Steudtner}}, \bibinfo {author} {\bibfnamefont {C.~S.}\ \bibnamefont {Tautermann}}, \bibinfo {author} {\bibfnamefont {M.}~\bibnamefont {Degroote}}, \bibinfo {author} {\bibfnamefont {N.}~\bibnamefont {Moll}}, \bibinfo {author} {\bibfnamefont {R.}~\bibnamefont {Santagati}},\ and\ \bibinfo {author} {\bibfnamefont {M.}~\bibnamefont {Streif}},\ }\href {https://doi.org/10.1103/PRXQuantum.5.010336} {\bibfield  {journal} {\bibinfo  {journal} {PRX Quantum}\ }\textbf {\bibinfo {volume} {5}},\ \bibinfo {pages} {010336} (\bibinfo {year} {2024})}\BibitemShut {NoStop}%
\bibitem [{\citenamefont {Su}\ \emph {et~al.}(2021)\citenamefont {Su}, \citenamefont {Berry}, \citenamefont {Wiebe}, \citenamefont {Rubin},\ and\ \citenamefont {Babbush}}]{su2021fault}%
  \BibitemOpen
  \bibfield  {author} {\bibinfo {author} {\bibfnamefont {Y.}~\bibnamefont {Su}}, \bibinfo {author} {\bibfnamefont {D.~W.}\ \bibnamefont {Berry}}, \bibinfo {author} {\bibfnamefont {N.}~\bibnamefont {Wiebe}}, \bibinfo {author} {\bibfnamefont {N.}~\bibnamefont {Rubin}},\ and\ \bibinfo {author} {\bibfnamefont {R.}~\bibnamefont {Babbush}},\ }\href@noop {} {\bibfield  {journal} {\bibinfo  {journal} {PRX Quantum}\ }\textbf {\bibinfo {volume} {2}},\ \bibinfo {pages} {040332} (\bibinfo {year} {2021})}\BibitemShut {NoStop}%
\bibitem [{\citenamefont {Fulton}\ and\ \citenamefont {Harris}(2013)}]{fulton2013representation}%
  \BibitemOpen
  \bibfield  {author} {\bibinfo {author} {\bibfnamefont {W.}~\bibnamefont {Fulton}}\ and\ \bibinfo {author} {\bibfnamefont {J.}~\bibnamefont {Harris}},\ }\href@noop {} {\emph {\bibinfo {title} {Representation theory: a first course}}},\ Vol.\ \bibinfo {volume} {129}\ (\bibinfo  {publisher} {Springer Science \& Business Media},\ \bibinfo {year} {2013})\BibitemShut {NoStop}%
\bibitem [{\citenamefont {Sagan}(2013)}]{sagan2013symmetric}%
  \BibitemOpen
  \bibfield  {author} {\bibinfo {author} {\bibfnamefont {B.~E.}\ \bibnamefont {Sagan}},\ }\href@noop {} {\emph {\bibinfo {title} {The symmetric group: representations, combinatorial algorithms, and symmetric functions}}},\ Vol.\ \bibinfo {volume} {203}\ (\bibinfo  {publisher} {Springer Science \& Business Media},\ \bibinfo {year} {2013})\BibitemShut {NoStop}%
\bibitem [{\citenamefont {Babbush}\ \emph {et~al.}(2019)\citenamefont {Babbush}, \citenamefont {Berry}, \citenamefont {McClean},\ and\ \citenamefont {Neven}}]{babbush2019quantum}%
  \BibitemOpen
  \bibfield  {author} {\bibinfo {author} {\bibfnamefont {R.}~\bibnamefont {Babbush}}, \bibinfo {author} {\bibfnamefont {D.~W.}\ \bibnamefont {Berry}}, \bibinfo {author} {\bibfnamefont {J.~R.}\ \bibnamefont {McClean}},\ and\ \bibinfo {author} {\bibfnamefont {H.}~\bibnamefont {Neven}},\ }\href@noop {} {\bibfield  {journal} {\bibinfo  {journal} {npj Quantum Information}\ }\textbf {\bibinfo {volume} {5}},\ \bibinfo {pages} {92} (\bibinfo {year} {2019})}\BibitemShut {NoStop}%
\bibitem [{\citenamefont {McArdle}\ \emph {et~al.}(2022)\citenamefont {McArdle}, \citenamefont {Campbell},\ and\ \citenamefont {Su}}]{mcardle2022exploiting}%
  \BibitemOpen
  \bibfield  {author} {\bibinfo {author} {\bibfnamefont {S.}~\bibnamefont {McArdle}}, \bibinfo {author} {\bibfnamefont {E.}~\bibnamefont {Campbell}},\ and\ \bibinfo {author} {\bibfnamefont {Y.}~\bibnamefont {Su}},\ }\href@noop {} {\bibfield  {journal} {\bibinfo  {journal} {Physical Review A}\ }\textbf {\bibinfo {volume} {105}},\ \bibinfo {pages} {012403} (\bibinfo {year} {2022})}\BibitemShut {NoStop}%
\end{thebibliography}%
\bibliographystyle{apsrev4-2}

\onecolumngrid
\newpage

\appendix

\setcounter{equation}{0}

\renewcommand\theequation{A.\arabic{equation}}

\section{Symmetry Sector Scaling Analysis}

\subsection{Linear Combination of Hermitian Operators}

In this section, we derive the analytical scaling behavior of the symmetry-sector bound $\Delta_\mu/2$. For generality, we consider a Hermitian operator $\hat{H}$ that is expressed as a sum of Hermitian operators $\hat{A}_i$, such that $\hat{H}=\sum_i \hat{A}_i$, where it is assumed that each Hermitian operator is further decomposed as a sum of unitaries, $\hat{A}_i = \sum_k c_k^{(i)}\hat{U}^{(i)}_k$, with complex coefficients, $c_k^{(i)}$. By using Weyl's inequalities for extremal eigenvalues, it is possible to show that:
\begin{equation}
    \tfrac{1}{2}\Delta \leq  \sum_i \tfrac{1}{2}\Delta^{(i)},
    \label{incoherent_bound}
\end{equation}
where $\Delta/2$ is the total spectral bound for $\hat{H}$ and $\Delta^{(i)}/2$ is the spectral bound for $\hat{A}_i$. Assuming that $\hat{H}$ and $\hat{A}_i$ obey the same symmetries (for instance, defined by the symmetry transformation $\hat{U}_S$ such that $[\hat{H},U_S]=0$ and $[\hat{A}_i,U_S]=0$), it is then straightforward to show that a similar inequality exists within a specific symmetry sector $\mu$,
\begin{equation}
    \tfrac{1}{2}\Delta_\mu \leq  \sum_i \tfrac{1}{2}\Delta^{(i)}_\mu.
    \label{incoherent_symmetry_sector_bound}
\end{equation}
These inequalities define the relations between the so-called coherent spectral bounds ($\Delta/2$ and $\Delta_\mu/2$) and incoherent spectral bounds ($\sum_i \tfrac{1}{2}\Delta^{(i)}$ and $\sum_i \tfrac{1}{2}\Delta_\mu^{(i)}$) discussed in the main text. We elaborate on this point within the context of the electronic structure Hamiltonian in the following subsection.

\subsection{Application: electronic structure Hamiltonian}
To elucidate the scaling behavior of the electronic structure Hamiltonian, we use the standard Majorana representation \cite{von2021quantum,koridon2021orbital},
\begin{align}
    \hat{H} = \mathcal{E}_o + H_\mathrm{1body} + H_\mathrm{2body},
\end{align}
where the \emph{1-body} and \emph{2-body} Hamiltonian terms are written as,
\begin{align}
    H_\mathrm{1body} &= \frac{i}{2} \sum_{\substack{pq\\\sigma}}\kappa_{pq}\hat{\gamma}_{p\sigma,0}\hat{\gamma}_{q\sigma,1}, \label{1body} \\
    H_\mathrm{2body} &= - \frac{1}{8}\sum_{\substack{pqrs\\\sigma\tau}} g_{pqrs} \hat{\gamma}_{p\sigma,0}\hat{\gamma}_{q\sigma,1} \hat{\gamma}_{r\tau,0}\hat{\gamma}_{s\tau,1},
\end{align}
and the core energy ($\mathcal{E}_o$) and the dressed one-electron integral matrix ($\kappa_{pq}$) are given by,
\begin{align}
    \mathcal{E}_o &=  E_\mathrm{nuc}+\sum_p h_{pp} + \tfrac{1}{2}\sum_{pr}g_{pprr} - \tfrac{1}{2}\sum_{pr} g_{prrp}, \\
    \kappa_{pq} &= h_{pq}-\frac{1}{2}\sum_r g_{prrq} + \sum_r g_{pqrr} .
    \label{2body}
\end{align}
The standard LCU implementation of the electronic structure Hamiltonian proceeds by finding LCU decompositions of $H_\mathrm{1body}$ and $H_\mathrm{2body}$ independently. This is typically achieved by directly using the Jordan-Wigner transformation of the Majorana operators which leads to the so-called \emph{sparse} encoding technique of the electronic structure Hamiltonian. In the sparse technique, only the Hamiltonian matrix elements larger than a certain threshold are prepared on the quantum computer \cite{babbush2019quantum,lee2021even}. Alternatively, tensor compression techniques such as \emph{double factorization} or \emph{tensor hypercontraction} factorize $\kappa_{pq}$ and $g_{pqrs}$ independently, before performing a Jordan-Wigner transformation in the final step \cite{von2021quantum,lee2021even}. In either case, both of these approaches will be bounded by the incoherent spectral gap, $\Delta_\mathrm{incoh.}/2$, which obeys the inequality,
\begin{equation}
    \tfrac{1}{2}\Delta \leq \tfrac{1}{2}\Delta_\mathrm{incoh.} \equiv \tfrac{1}{2}\Delta^\mathrm{1body}+\tfrac{1}{2}\Delta^\mathrm{2body},
\end{equation}
referenced in the main text. This inequality also holds within each symmetry sector, $\mu$, given that $H_\mathrm{1body}$ and $H_\mathrm{2body}$ obey the same symmetries as $\hat{H}$. In the following, we proceed to find the symmetry-sector scaling behavior of the 1-body ($\Delta^\mathrm{1body}/2$) and 2-body ($\Delta^\mathrm{2body}/2$) Hamiltonian spectral gaps. 

\subsubsection{1-body Hamiltonian}
The 1-body spectral gap can be determined exactly. The result is obtained by using the eigen-decomposition of the one-electron integral matrix,
\begin{equation}
    \kappa_{pq} = \sum_k \lambda_k U_{pk}U_{qk}.
\end{equation}
This leads to the following expression for the 1-body Hamiltonian,
\begin{equation}
    H_{\mathrm{1body}} = \frac{i}{2}\sum_{\substack{pq\\ \sigma}}\kappa_{pq}\hat{\gamma}_{p\sigma,0} \hat{\gamma}_{q\sigma,1} = \frac{i}{2}\hat{G}^\dagger\left(\sum_{k,\sigma} \lambda_k \hat{\gamma}^\dagger_{k\sigma,0} \hat{\gamma}_{k\sigma,1} \right) \hat{G},
\end{equation}
where $\hat{G} =  \exp\left(\sum_{p,q,\sigma}[\log \mathbf{U}]_{pq} \hat{a}^{\dagger}_{p\sigma}\hat{a}_{q\sigma}  \right)$ is the unitary orbital rotation (Givens) operator. Assuming the eigenvalues $\lambda_k$ are ordered from smallest to largest, the symmetry-sector bound of the 1-body Hamiltonian with fixed particle number, $\mu=(\eta_\alpha,\eta_\beta$), is given by:
\begin{align}
    \frac{1}{2}\Delta_\mu^{\mathrm{1body}} = \frac{1}{2}\left( \sum_{\substack{k=N_\mathrm{orb}-\eta_\sigma \\ \sigma}}^{N_\mathrm{orb}-1} \!\!\!\! \lambda_k - \sum_{ \substack{k=0 \\ \sigma} }^{\eta_\sigma-1} \lambda_k \right).
    \label{1body_bound}
\end{align}
As a point of reference, we also compute the fermionic semi-norm, $\| \hat{X}  \|_\mu :=\max_{\ket{\psi_\mu}} |\braket{\psi_\mu|\hat{X}|\psi_\mu} |$, defined for a general Hermitian operator $\hat{X}$ \cite{mcardle2022exploiting}. For the 1-body Hamiltonian, we find:
\begin{equation}
    \|H_{\mathrm{1body}} \|_\mu =  \sum_\sigma \text{max}\left[ \Bigg|\frac{1}{2}\!\!\sum_{k=0}^{N_\mathrm{orb}-1} \!\!\!\lambda_k- \sum_{ \substack{k=0 } }^{\eta_\sigma-1} \lambda_k\Bigg| ,\Bigg| \frac{1}{2}\!\!\sum_{k=0}^{N_\mathrm{orb}-1} \!\!\!\lambda_k - \!\!\!\!\!\!\sum_{\substack{k=N_\mathrm{orb}-\eta_\sigma }}^{N_\mathrm{orb}-1} \!\!\!\! \lambda_k \Bigg|\right]  ,
    \label{1body_norm}
\end{equation}
which is equal to the largest singular value of $H_\mathrm{1body}$ in the symmetry sector $\mu=(\eta_\alpha,\eta_\beta$). The symmetry sector spectral gap \eqref{1body_bound} is a lower bound to the fermionic semi-norm, $\frac{1}{2}\Delta_\mu^{\mathrm{1body}} \leq \| H_{\mathrm{1body}} \|_\mu$, which is not saturated for the 1-body Hamiltonian expressed in the form above. It is important to note that these results provide a simple and direct way of computing the spectral gaps exactly for 1-body Hamiltonian based on $\{\hat{N},\hat{S}_z\}$ symmetries. To elucidate the scaling behavior, we use the following bound
\begin{equation}
    \| H_{\mathrm{1body}} \|_\mu \leq \|\kappa_{pq}\| (\eta_\alpha + \eta_\beta) = \text{max}[|\Lambda|]\eta,
\end{equation}
where $\Lambda = \{\lambda_1,\lambda_2,\cdots,\lambda_{N_\mathrm{orb}}\}$ denotes the set of eigenvalues of the one-electron integral matrix. This shows that the symmetry sector gap will be strictly upper bounded by the largest absolute eigenvalue of the one-electron integral matrix as well as the total number of electrons, $\eta = \eta_\alpha + \eta_\beta$.

\subsubsection{2-body Hamiltonian}
To derive the analytical scaling behavior of the two-body Hamiltonian, we use the \emph{double factorized} representation of the four-index tensor \cite{von2021quantum,cohn2021quantum},
\begin{equation}
    g_{pqrs} = \sum_{tkl} \alpha^t_k\alpha ^t_l \;U^{(t)}_{pk}U^{(t)}_{qk}U^{(t)}_{rl}U^{(t)}_{sl},
\end{equation}
which transforms the two-body Hamiltonian \eqref{2body} into the complete-square form,
\begin{equation}
    H_\mathrm{2body} = \frac{1}{2}\sum_t \hat{v}^2_t 
\end{equation}
where 
\begin{equation}
    \hat{v}_t = \hat{G}_t^\dagger  \left(\frac{i}{2}\sum_{\substack{k\\\sigma}} \alpha^t_k  \hat{\gamma}_{k\sigma,0}\hat{\gamma}_{k\sigma,1} \right) \hat{G}_t,
\end{equation}
and the orbital rotation operator is defined for each $t$-leaf as, $\hat{G}_t =  \exp\left(\sum_{p,q,\sigma}[\log \mathbf{U}^{(t)}]_{pq} \hat{a}^{\dagger}_{p\sigma}\hat{a}_{q\sigma}  \right)$. The spectral bound for the two-body term is then given by,
\begin{equation}
    \tfrac{1}{2}\Delta_\mu^{\mathrm{2body}} \leq \tfrac{1}{4}\sum_t \Delta_\mu^{(v^2_t)}
\end{equation}
where 
\begin{align}
    \Delta_\mu^{(v^2_t)} &= (E_\mathrm{max}^{(\mu,v^2_t)} - E_\mathrm{min}^{(\mu,v^2_t)}) \\
    &\leq \| \hat{v}_t^2 \|_\mu 
    = \| \hat{v}_t \|_\mu^2.
\end{align}
Using the results from the previous subsection for the 1-body term, 
\begin{equation}
    \|\hat{v}_t \|_\mu =  \sum_\sigma \text{max}\left[ \Bigg|\frac{1}{2}\sum_{k=0}^{N_\mathrm{orb}-1}\alpha^t_k- \sum_{ \substack{k=0 } }^{\eta_\sigma-1} \alpha^t_k\Bigg| ,\Bigg| \frac{1}{2}\sum_{k=0}^{N_\mathrm{orb}-1} \alpha^t_k - \!\!\!\!\!\!\sum_{\substack{k=N_\mathrm{orb}-\eta_\sigma }}^{N_\mathrm{orb}-1} \!\!\!\! \alpha^t_k \Bigg|\right]  ,
\end{equation}
we find the following upper bound to the scaling behavior of the symmetry sector bound,
\begin{align}
    \tfrac{1}{2}\Delta_\mu^{\mathrm{2body}} \leq \tfrac{1}{4}\sum_t \|\hat{v}_t\|^2_\mu \leq \frac{1}{4}\sum_t \mathcal{O}\left(g(N_\mathrm{orb})\eta^2 \right),
\end{align}
in agreement with similar scaling arguments derived by McArdle \emph{et al.} \cite{mcardle2022exploiting}. This result shows that the symmetry sector bound will be strictly upper bounded by a function $g(N_\mathrm{orb})$ dependent on the number of spatial orbitals as well as the square of the total number of electrons, $\eta^2$. This result is highlighted in the analysis and numerical results section of the main manuscript. It is important to note that this upper bound can be quite loose, and a more accurate assessment of the scaling behavior is obtained by performing numerical calculations of the spectral bounds as discussed in the following section. 

\section{Numerical limitations of orbital optimization procedure}
While the orbital optimization method provides a tractable approach to compute symmetry-aware spectral bounds for large-scale systems, it does have important limitations and challenges which we discuss below. First, as alluded in the main manuscript, we found that the numerical procedure was prone to certain convergence issues arising from the highly non-convex landscape encountered in the orbital optimization procedure. In particular, we found that random initial guesses highly affected final converged energy values. To address this issue, we proposed the pre-orthogonalized one-electron and two-electron integrals in Eqs. \eqref{hij} and \eqref{gijkl}. We also found that initializing the $N_\mathrm{orb}(N_\mathrm{orb}-1)/2$ parameters to zero helped in obtaining reliable solutions for many cases. Nevertheless, this is not fool-proof and there is no guarantee that a global energy minimum will always be found. Second, and perhaps more importantly, it is crucial to emphasize that the proposed method is a single-determinant method which does not provide a complete description of multi-reference quantum states. As a result, it cannot accurately describe physical states such as open-shell singlets, which are important for the description of certain chemical ground-states. While this could, in principle, lead to a large deviation between $\Delta^{(\mathrm{HF})}$ and $\Delta^{(\mathrm{FCI})}$, it is not clear how large the deviation will be in practice. Below, we provide empirical evidence for the efficacy of the orbital optimization approach for various small-scale systems, comparing with FCI and DMRG-based numerical approaches. 

It is important to note that despite these two limitations, our method still adheres to the variational principle, ensuring that the estimated spectral range for each symmetry sector remains a lower bound to the exact FCI solution, $\Delta^{(\mathrm{HF})}_\mathrm{local\,min.} \leq \Delta^{(\mathrm{HF})}_\mathrm{global\,min.} \leq \Delta^{(\mathrm{FCI})}$ (see Figure \ref{fig:variational_methods}). This remains true regardless of whether the FCI extremal energies require a multi-configurational description, or whether the orbital optimization procedure converges to a bad local minima. In Figure \ref{fig:FCI_vs_HF_SI}, we present numerical results which validate this claim for the linear Hydrogen chain with the number of atoms ranging from 4 to 16 expanded in the minimal STO-6G basis. Table \ref{tab:results_sto6g} and \ref{tab:results_631g} similarly provide a comparison between the orbital-optimized spectral bounds with FCI-predicted spectral bounds for a small-molecule benchmark set expanded in the STO-3G and 6-31G basis. Table \ref{tab:results_DMRG} also provides a comparison between DMRG results and Hartree-Fock predictions for additional benchmark systems including stretched nitrogen, water, as well as the proposed (33e, 31o) active space for P450  \cite{goings2022reliably}.
\begin{figure}[t!]
    \subfloat[\label{fig:sub1}]{%
      \includegraphics[width=.45\linewidth]{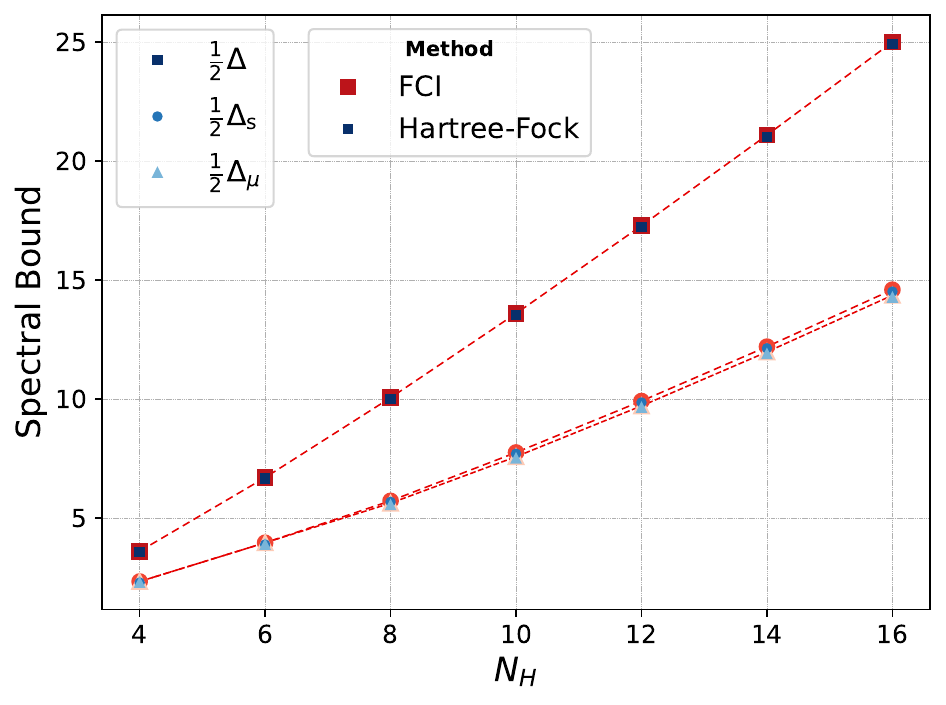}%
    }
    \subfloat[\label{fig:sub2}]{%
      \includegraphics[width=.45\linewidth]{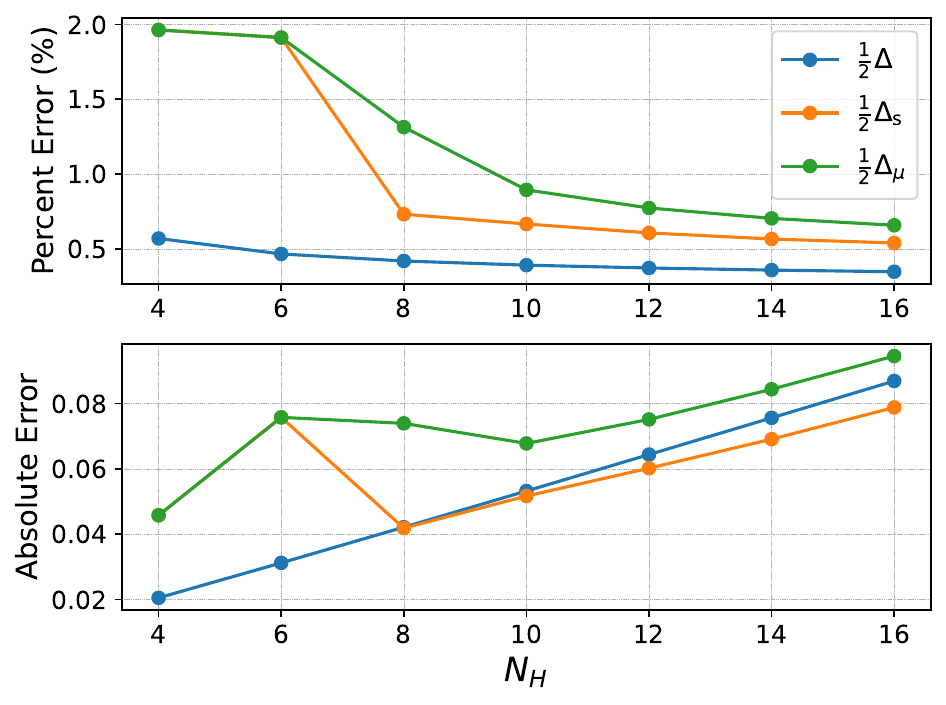}%
    }\hfill
    \caption{ (a) Numerical comparison between FCI spectral bounds (red markers) and orbital-optimized spectral bounds (blue markers) reported in units of Hartree. Symmetry-aware spectral bounds are plotted as a function of the number of hydrogen atoms, $N_H$. (b) Percent error and absolute error are also provided as a function of hydrogen atoms, $N_H$.}
    \label{fig:FCI_vs_HF_SI}
\end{figure}

\begin{table*}[h]
\centering
\begin{tabular}{|c|a|b|a|b|a|}
\hline
System   & Bound & Total (incoh. HF) & Total (incoh. FCI)  & Total (coh. HF) & Total (coh. FCI)   \\ \hline
\multirow{3}{*}{BeH$_2$} & $\!\!\Delta/2$ & 13.4 & 13.4 & 10.0 & 10.0 \\
                        & $\Delta_s/2$    & 8.3  & 8.4 & 7.2 & 7.3 \\
                        & $\Delta_\mu/2$  & 7.9  & 8.3 & 7.2 & 7.3 
\\ \hline \hline   
\multirow{3}{*}{H$_2$O} & $\!\!\Delta/2$ & 48.7 & 48.7 & 41.9 & 41.9 \\
                        & $\Delta_s/2$   & 29.6 & 30.0 & 28.9 & 28.9\\
                        & $\Delta_\mu/2$ & 29.4 & 29.9 & 23.7 & 23.7 
\\ \hline \hline    
\multirow{3}{*}{NH$_3$} & $\!\!\Delta/2$ & 38.5 & 38.5  & 33.8 & 33.8 \\
                    & $\Delta_s/2$       & 23.9 & 24.3   & 22.6 & 22.6 \\
                    & $\Delta_\mu/2$     & 23.3 & 23.9   & 19.4 & 19.5
\\ \hline   
\end{tabular}
\caption{Comparison between orbital-optimized (HF) and FCI spectral bounds in the minimal STO-6G basis reported in units of Hartree.}
\label{tab:results_sto6g}
\end{table*}

\begin{table*}[b]
\centering
\begin{tabular}{|c|a|b|a|b|a|}
\hline
System   & Bound & Total (incoh. HF) & Total (incoh. FCI)  & Total (coh. HF) & Total (coh. FCI)   \\ \hline
\multirow{3}{*}{BeH$_2$} & $\!\!\Delta/2$ & 42.1 & 42.1 & 32.6 & 32.6 \\
                        & $\Delta_s/2$    & 14.5 & 14.9 & 11.1 & 11.2 \\
                        & $\Delta_\mu/2$  & 13.0 & 13.4 & 9.3 & 9.3 
\\ \hline \hline   
\multirow{3}{*}{H$_2$O} & $\!\!\Delta/2$ & 57.4 & 57.4 & 42.4 & 42.5 \\
                        & $\Delta_s/2$   & 38.3 & 38.9 & 34.4 & 34.5 \\
                        & $\Delta_\mu/2$ & 37.0 & 38.0 & 34.3 & 34.4 
\\ \hline \hline    
\multirow{3}{*}{NH$_3$} & $\!\!\Delta/2$ & 65.2 & 65.2 & 43.5 & 43.6 \\
                    & $\Delta_s/2$       & 31.8 & 32.8 & 27.3 & 27.4 \\
                    & $\Delta_\mu/2$     & 30.6 & 31.6 & 27.3 & 27.4 
\\ \hline   
\end{tabular}
\caption{Comparison between orbital-optimized (HF) and FCI spectral bounds expanded in the 6-31G basis reported in units of Hartree.}
\label{tab:results_631g}
\end{table*}

\begin{table*}[h]
\centering
\begin{tabular}{|c|c|c|c|}
\hline
System   & Spectral Bound & Total (coh. HF) & Total (coh. DMRG)   \\ \hline
\multirow{2}{*}{P450 (33e,31o)} & \multirow{2}{*}{$\Delta_\mu/2$} & \multirow{2}{*}{38.4} & \multirow{2}{*}{38.5} \\
& & & 
\\ \hline \hline    
\multirow{2}{*}{N$_2$ [cc-pVDZ]} & \multirow{2}{*}{$\Delta_\mu/2$} & \multirow{2}{*}{55.6} & \multirow{2}{*}{55.9} \\
& & & 
\\ \hline \hline
\multirow{2}{*}{H$_2$O [cc-pVDZ]} & \multirow{2}{*}{$\Delta_\mu/2$} & \multirow{2}{*}{42.8} & \multirow{2}{*}{42.9} \\
& & & 
\\ \hline \hline    
 \multirow{2}{*}{H$_2$O [cc-pVTZ]} & \multirow{2}{*}{$\Delta_\mu/2$} & \multirow{2}{*}{63.44} & \multirow{2}{*}{63.95} \\
& & & 
\\ \hline 
\end{tabular}
\caption{Comparison between orbital-optimized (HF) and DMRG spectral bounds reported in units of Hartree.}
\label{tab:results_DMRG}
\end{table*}

\end{document}